\documentclass{article}
\usepackage{graphicx} 
\usepackage{amsmath}
\usepackage{amssymb}
\usepackage{color}
\usepackage{url}
\usepackage{algorithm}
\usepackage[noend]{algpseudocode}

\usepackage{todonotes}
\usepackage{amsthm}
\newtheorem{theorem}{Theorem}[section]
\newtheorem{proposition}{Proposition}[section]

\newtheorem{definition}{Definition}[section]

\title{Asymptotic Counting of Binary Phylogenetic Networks}
\author{Hao Yu, Louxin Zhang\thanks{Corresponding author: Department of Mathematics, National University of Singapore, 10 Lower Kent Ridge Road, Singapore 119076. Email: matzlx@nus.edu.sg.}\\
Department of Mathematics\\
National University of Singapore}
\date{}

\begin{document}

\maketitle

\begin{abstract}
Phylogenetic networks provide a powerful framework for modeling reticulate evolutionary processes such as hybridization, recombination, and horizontal gene transfer.  Using edge insertion, we analyze the local structures that affect the number of possible constructions of binary phylogenetic networks. By bounding the contribution of networks with exceptional local configurations and combining these bounds with known asymptotic formulas for tree-child networks, we show that, when $k=o(\sqrt n)$, the number of binary phylogenetic networks with $k$ reticulations on $n$ taxa is asymptotic to
\[
\binom{n}{k}2^{n+k-1/2}n^{n+k-1}e^{-n}.
\]
Using the exact count of tree-child networks, we further show that normal networks have the same asymptotic count. This answers two open problems in counting phylogenetic networks.
\end{abstract}

\section{Introduction}

Phylogenetic trees have been used as the standard model for representing evolutionary relationships among species for several decades \cite{doolittle1999}. 
In a tree model, evolutionary history is assumed to be vertically divergent, with each lineage descending from a unique ancestral lineage. 
However,  inherently reticulate evolutionary processes exist widely, including hybridization, recombination, horizontal gene transfer, and introgression. 
These processes cannot be adequately represented by a phylogenetic tree, and this has motivated the study of phylogenetic networks as a more general framework for modeling evolutionary histories with reticulation events in genomes
\cite{Gusfield_book,HusonRuppScornavacca2010,Steel_book,zhang2019clusters}.

From a mathematical point of view, phylogenetic networks form a much richer and more complex class of structures than phylogenetic trees. The space of phylogenetic networks is huge and therefore inference of phylogenetic networks is much more challenging \cite{zhang2026phylofusion}.
While the number of binary phylogenetic trees on a fixed set of $n$ taxa is given by $(2n-3)!!$, the enumeration of phylogenetic networks is substantially more difficult 
\cite{Semple_15}. 
The difficulty arises from the fact that networks on the same number of taxa and reticulations may have very different local and global structures 
\cite{HaoYu2026_JCB}. 
Understanding such networks through counting is important for phylogenetic inference, because the size of the underlying search space directly affects the design, analysis, and feasibility of reconstruction algorithms.

To make the study of phylogenetic networks more tractable, several biologically and mathematically meaningful subclasses have been introduced. 
 These include tree-child networks \cite{CardonaRosselloValiente2009}, 
   galled trees \cite{Gusfield_04,Wang_01}, galled networks \cite{HusonRuppScornavacca2010}, and tree-based networks \cite{Francis_15}. 
Among these classes, tree-child networks have received particular attention. 
They exclude certain degenerate configurations by requiring every non-leaf node to have at least one child that is a tree node or a leaf. 
This condition leads to favorable combinatorial and algorithmic properties, while still allowing a broad range of reticulate evolutionary scenarios.
Notably, there is an elegant recurrence formula  for counting binary tree-child networks, which was conjectured by Pons and Batle \cite{Conjecture_Pons} and proved by Liu et al. \cite{Conjecture_Proof} recently. Counting for tree-child networks have been extensively studied \cite{bouvel2020counting,cardona2020counting,chang2024,Fuchs_18,fuchs2022counting,Zhang_19}.

In this paper, we study the asymptotic enumeration of binary phylogenetic networks with $k$ reticulations on $n$ taxa, with particular emphasis on the range in which $k$ is non-constant but small relative to $n$. 
The enumeration of binary phylogenetic networks was first studied by McDiarmid, Semple, and Welsh \cite{Semple_15}. 
For fixed $k$, the leading asymptotic term  for the number of binary phylogenetic networks with $k$ reticulations on $n$ taxa is proved to be ${n\choose k} 2^{n+k-1/2}n^{n+k-1}e^{-n}$ in \cite{mansouri2022counting}; see also \cite{HaoYu2026_JCB}. 
The same asymptotic formula also holds for tree-child networks \cite{fuchs2022counting}. 
Our aim is to extend this analysis to a broader range for $k$ using a new technique by showing that, when $k=o(\sqrt{n})$, the same leading asymptotic term holds for counting binary phylogenetic networks. 

We analyze how networks can be generated by edge insertion and identify the structural configurations that affect the number of possible constructions. 
By bounding the contribution of networks with exceptional local structures and combining these bounds with known asymptotic formulas for tree-child networks, we derive asymptotic counting results for the class of binary phylogenetic networks. 
Our results show that, when $k=o(\sqrt{n})$, the same leading asymptotic term holds for counting binary phylogenetic networks and for counting normal networks. This settles open problems posed in \cite{fuchs2024asymptotic}.

\section{Basic Concepts and Notation}

In this section, we introduce the basic definitions and notation used throughout the paper.

\subsection{Binary Phylogenetic Networks}

Let $X$ be a finite set of taxa. A \emph{binary phylogenetic network} (BPN) on $X$ is a rooted directed acyclic graph (DAG) $N$ satisfying the following properties:
\begin{itemize}
    \item The root has indegree $0$ and outdegree $1$.
    \item Each leaf has indegree $1$ and outdegree $0$.  The leaves are bijectively labeled by the elements of $X$.
    \item Every internal node is either:
    \begin{itemize}
        \item a \emph{tree node} with indegree $1$ and outdegree $2$, or
        \item a \emph{reticulation (node)} with indegree $2$ and outdegree $1$.
    \end{itemize}
\end{itemize}
The set of nodes and directed edges in $N$ are denoted by $V(N)$ and $E(N)$, respectively. A BPN is given in Figure~\ref{fig:NtkExamples}a, where the network root is not drawn.

Let $u$ and $v$ be two nodes in $N$. The node $v$ is said to be \emph{below} $u$ if there is a directed path from $u$ to $v$. If $v$ is below $u$, we also say that
$u$ is an ancestor of $v$ and $v$ is a descendant of $u$. 

An edge entering a reticulation node is called a \emph{reticulation edge}, while all other edges are called \emph{tree edges}. The number of reticulations in $N$ is denoted by $r(N)$.

We summarize basic properties of BPNs in the following proposition, which will be used without further mention.

\begin{proposition}
Let $N$ be a BPN with $k$ reticulations on a set of $n$ taxa. Then:
\begin{enumerate}
\item $N$ has $n + k - 1$ tree nodes.

\item $N$ has $2n + 3k - 1$ edges, of which $2n + k - 1$ are tree edges and $2k$ are reticulation edges.
\end{enumerate}
\end{proposition}

\begin{figure}[!b]
    \centering
    \includegraphics[width=0.7\textwidth]{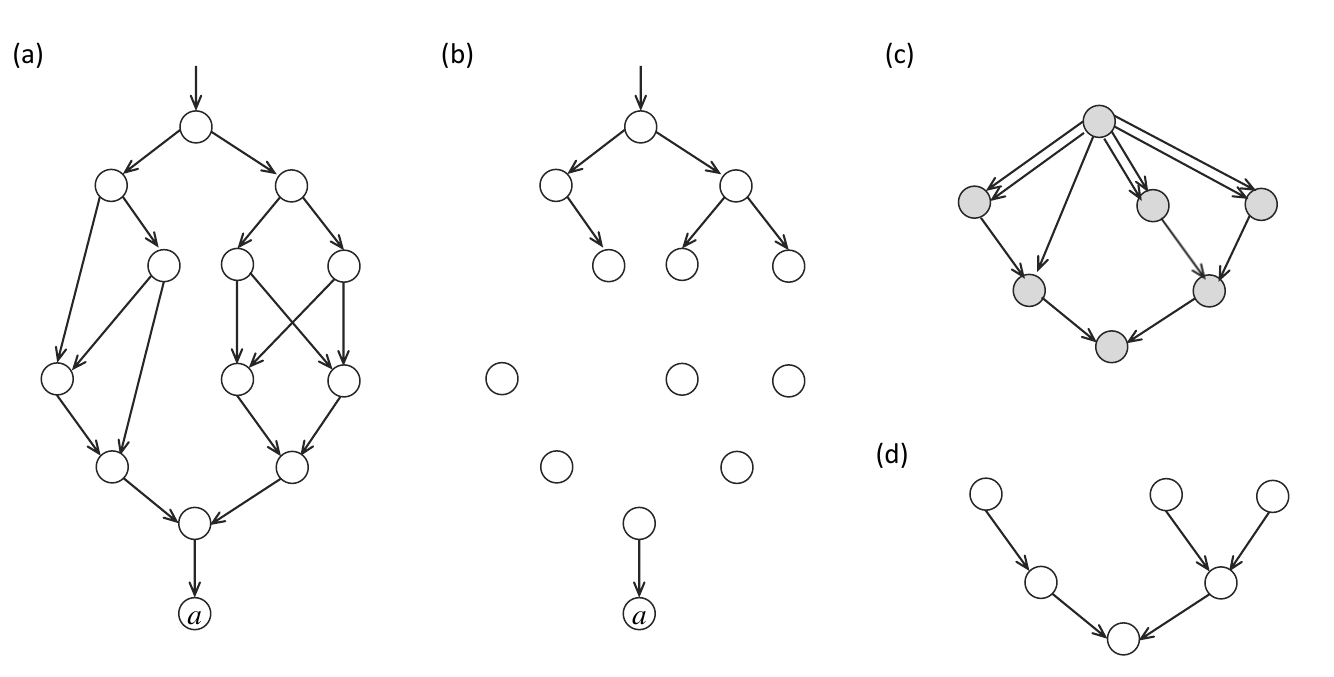}
    \caption{ (a) A binary phylogenetic network $N$ on one taxon $a$. (b) All tree components of $N$. (c) The component graph of $N$.
    (d) The unique reticulation cluster of $N$.
    \label{fig:NtkExamples}
    }
\end{figure}

Let $\Sigma$ be a set of taxa. A \textit{binary phylogenetic tree} (BPT) on $\Sigma$ is a BPN on $\Sigma$ that contains no reticulation nodes.

A \textit{forest} on $\Sigma$ is a disjoint union of BPTs $T_i$ defined on taxon sets $\Sigma_i$, where $\{\Sigma_i\}$ forms a partition of $\Sigma$, that is, $\bigcup_i \Sigma_i = \Sigma$ and $\Sigma_i \cap \Sigma_j = \emptyset$ for $i \neq j$.

\begin{proposition} 
\label{proposition4}
(Proposition 2.8.1, \cite{semple_book})
Let $\Sigma$ be a set of $n$ taxa, and let ${\cal F}_{n,k}$ denote the set of forests on $\Sigma$ consisting of $k$ rooted BPTs, where $n \ge k \ge 1$. Then,
\begin{eqnarray}
|{\cal F}_{n, k}|=\frac{(2n-k-1)!}{2^{\,n-k}(n-k)!(k-1)!}.
\end{eqnarray}
In particular, the number $t_n$ of BPTs on $\Sigma$ is
\begin{eqnarray}
t_n = |{\cal F}_{n,1}| = \frac{(2n-2)!}{2^{\,n-1}(n-1)!}.
\label{eqn22222}
\end{eqnarray}
\end{proposition}

\subsection{Tree Components and Component Graphs}

Let $N$ be a BPN. The \emph{tree components} of $N$ are defined as the connected components of the subgraph obtained by deleting all reticulation edges from $N$.
The subgraph obtained in this way has $r(N)+1$ connected components, each of which is a tree, possibly consisting of a single vertex, as shown in Figure~\ref{fig:NtkExamples}b.
We note that the definition of tree components adopted here differs slightly from those in earlier works \cite{Gambette_15,Gunawan_16,zhang2019clusters}, in order to better suit our purposes.


We distinguish the \emph{top tree component}, which has the same root as the given network $N$, from the remaining tree components, which are rooted at a reticulation node.


Let $\mathcal{C}$ be the set of tree components of $N$. The \emph{component graph} of $N$, denoted by $\mathcal{G}(N)$, is the directed graph defined as follows:
\begin{itemize}
    \item Each vertex of $\mathcal{G}(N)$ corresponds to a tree component in $\mathcal{C}$.
    \item There is a directed edge from a component $C_1$ to a component $C_2$ if and only if the root of $C_2$  is a child of a tree node lying in $C_1$.
\end{itemize}

The component graph $\mathcal{G}(N)$ is a directed acyclic graph (DAG), and it captures the hierarchical relationships among tree components induced by reticulation edges, shown in Figure~\ref{fig:NtkExamples}c. Note that there may be two parallel edges between two nodes in $\mathcal{G}(N)$.



\subsection{Reticulation clusters}

Let $N$ be a BPN. The \emph{reticulation clusters} of $N$ are defined as the connected components of
the subgraph obtained by deleting all tree nodes  (and their incident edges) from $N$. 
Each reticulation cluster can be viewed as an in-arborescence. 

\begin{proposition}
\label{prop3_entering_ret_cluster}
Let $N$ be a BPN, let   $R$ be a reticulation cluster with $a$ ($a\geq 1$) reticulation nodes. Then,
  
  (1) There are $a-1$ reticulation edges between the reticulation nodes in $R$,

  (2) There are $a+1$ edges that enter the cluster $R$ from tree nodes (i.e., edges
from tree nodes into reticulation nodes at the top of $R$).

  (3) The lowest reticulation node has exactly one outgoing edge, and its child is either a tree node or a leaf.
 \end{proposition}

 \begin{proposition}
  Let $N$ be a BPN with $k$ reticulations.
  If $N$ has $c$ reticulation clusters, then 

  (1) There are $k+c$ edges entering reticulation clusters from tree nodes.

  (2) There are $k-c$ edges between the reticulation nodes.


\end{proposition}

\subsection{Tree-Child Networks}

A BPN $N$ is called a \emph{tree-child network} if every internal node has at least one child that is a tree node or a leaf. Equivalently, each internal node is connected to a leaf via a path consisting of tree edges.

This condition ensures that from every internal node there exists a directed path to a leaf that avoids reticulation nodes except possibly at the starting vertex. It also implies that 
the leaves of each tree component are leaves of $N$ if $N$ is tree-child.

A tree-child network is said to be a \textit{simplex} network if the child of every reticulation node is a leaf. 

\begin{proposition}
\label{proposition5_complex}
Let $\Sigma=\{1, 2, ..., n\}$. The number of simplex networks with $k$ reticulations on $\Sigma$ in which the children of $k$ reticulations are 
leaves 1 to $k$ is $\frac{(2n - 2)!}{2^{\,n-1}(n - k - 1)!}$.
\end{proposition}

A tree-child network is said to be {\it normal} if there is no directed path from one parent to another for each reticulation node.




\begin{figure}[!t]
    \centering
    \includegraphics[width=0.7\textwidth]{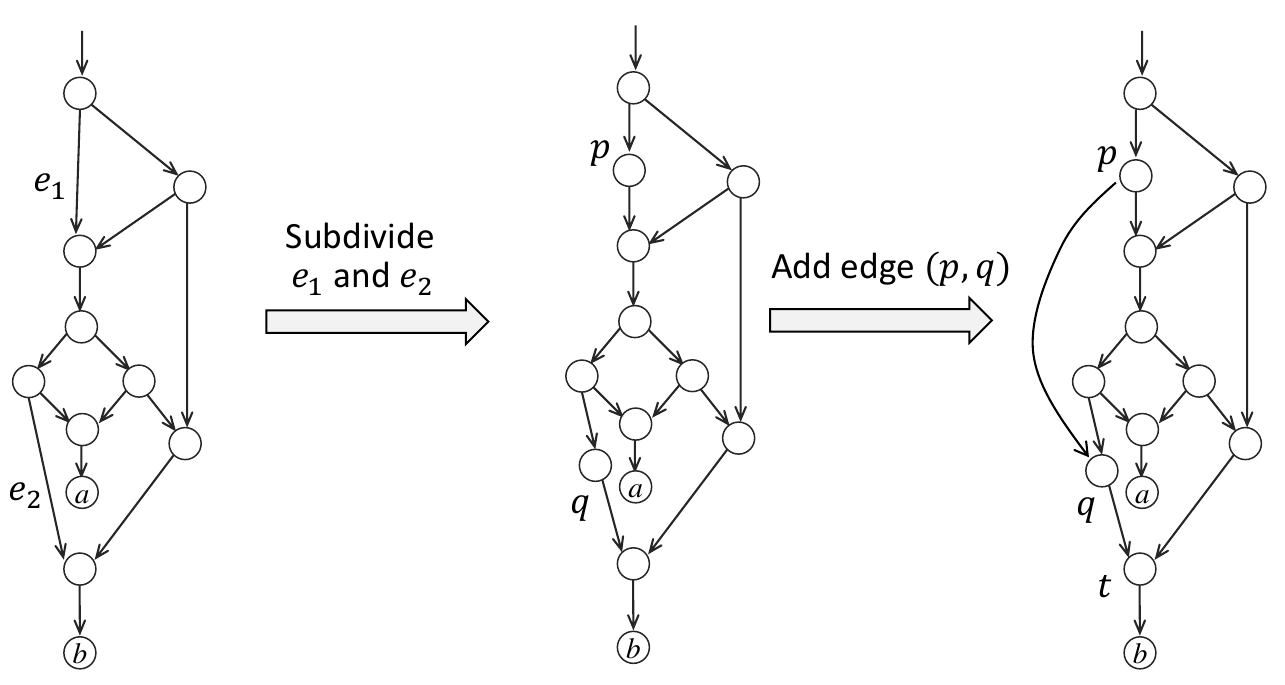}
    \caption{ Illustration of edge insertion in a phylogenetic network.
    \label{fig:Edge_insertion}
    }
\end{figure}

\section{Generating BPNs by Edge Insertion}

\subsection{Edge Insertion}

If $e=(u,v)$ is a directed edge in a directed graph, \emph{edge subdivision} replaces it by two edges:
$$(u,w), (w,v),$$
preserving the direction. If $w$ is a node of indegree 1 and outdegree 1 in a directed graph,  \emph{contraction} replaces the unique incoming edge $(u, w)$ and the unique outgoing edge $(w, v)$ of $w$ with 
$(u, v)$.

Consider a phylogenetic network $N$ with $k-1$ reticulation nodes. Let $e_1=(u, v)$ and 
$e_2=(s, t)$ be two edges of $N$ such that 
$t$ is not an ancestor of $u$. 
We construct a phylogenetic network $N'$ by inserting a directed edge from $e_1$ to $e_2$ as follows:
\begin{itemize}
\item Subdivide $e_1$ into two edges $(u, p)$ and $(p, v)$,
\item Subdivide $e_2$ into two edges $(s, q)$ and $(q, t)$, and
\item Add a directed edge $(p, q)$.
\end{itemize}
The resulting $N'$ is a network with $k$ reticulation nodes on the same set of taxa. This edge insertion operation is illustrated  in 
Figure~\ref{fig:Edge_insertion}.  In the rest of this paper,  $N'$ is said to be a network obtained from $N$ by inserting an edge ``from $e_1$ to $e_2$".

Assume that a BPN $N$ can be obtained from $N'$ by inserting an edge from $e'_1$ to $e'_2$, 
and also from $N''$ by inserting an edge from $e''_1$ to $e''_2$. 
These two edge-insertion operations are said to be \emph{different} if at least one of the following conditions holds:
\begin{itemize}
\item $N'$ and $N''$ are distinct.
\item $e'_1 \neq e''_1$.
\item $e'_2 \neq e''_2$.
\end{itemize}

A BPN with $k \, (\geq 1)$ reticulations can typically be obtained from a BPN with $k-1$ reticulations via edge insertion in multiple ways. 
However, a small subset of  BPNs cannot be constructed in this manner.  Such a BPN is shown in Figure~\ref{fig:symmetric_Case}c.


\begin{figure}[htbp]
    \centering
    \includegraphics[width=0.7\textwidth]{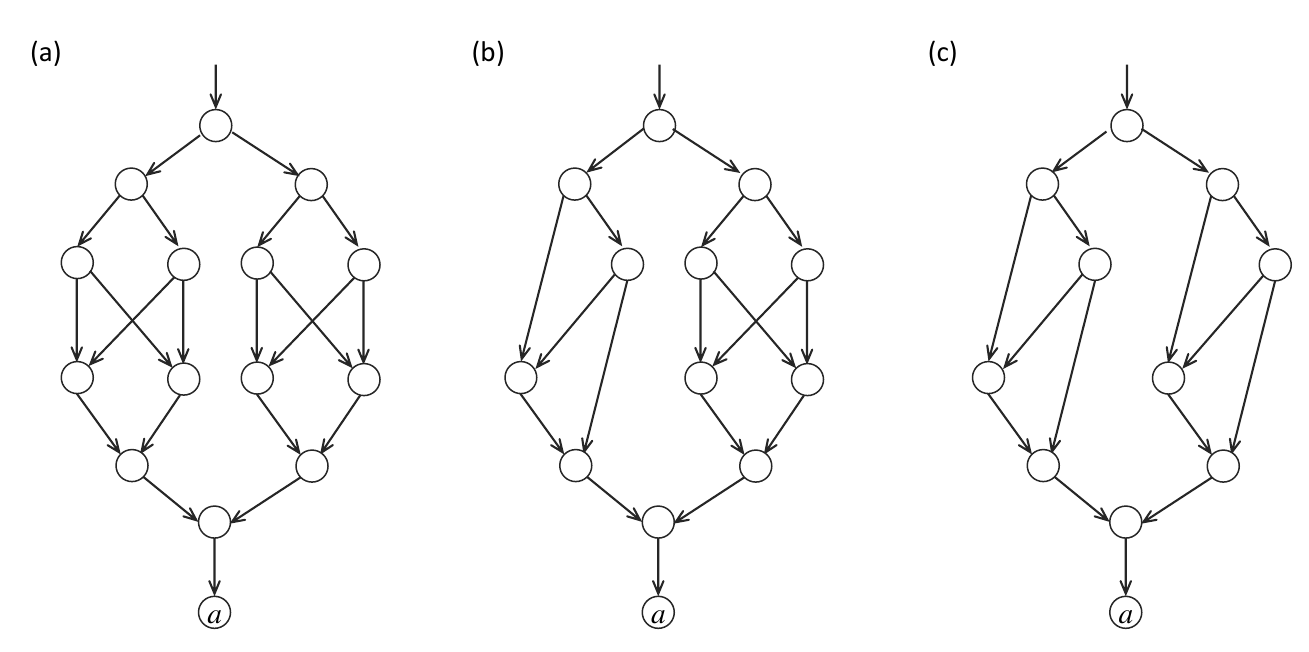}
    \caption{Special cases. (a) A highly symmetric BPN, which can only be generated by edge insertion from the BPN in (b) in one way. (c) A  BPN that cannot be generated by edge insertion.
    \label{fig:symmetric_Case}
    }
\end{figure}

\subsection{Conjugate Edges}

Let $N$ be a BPN with $k$ reticulations. 
Each reticulation node has two parents, not necessarily tree nodes.
Let $e_i=(p_i, r_i)$ ($1 \leq i \leq t$) be all edges of $N$ from tree nodes to reticulation nodes.
It is easy to see that $N$ can be obtained from a BPN with $k-1$ reticulations via edge insertion in at most $t$ ways.

If $N$ cannot be obtained in $t$ distinct ways by edge insertion, then at least one of the following two cases occurs.

\begin{enumerate}
\item There exists an edge $e_j$ such that deleting $e_j$ and contracting its endpoints does not result in a BPN (see Figure~\ref{fig:symmetric_Case}c).

\item There exist two distinct edges $e_a=(p_a,r_a)$ and $e_b=(p_b,r_b)$ such that deleting $e_a$ and contracting $p_a$ and $r_a$ yields the same BPN $N'$ with $k-1$ reticulations as deleting $e_b$ and contracting $p_b$ and $r_b$. Conversely, $N$ can be obtained from $N'$ by inserting an edge between a unique pair of edges, as illustrated in Figures~\ref{fig:DeletionC1} and~\ref{fig:DeletionC2}.
\end{enumerate}

In the rest of this subsection, we introduce concepts that will be useful for studying the second case.

\begin{figure}[!t]
    \centering
    \includegraphics[width=0.9\textwidth]{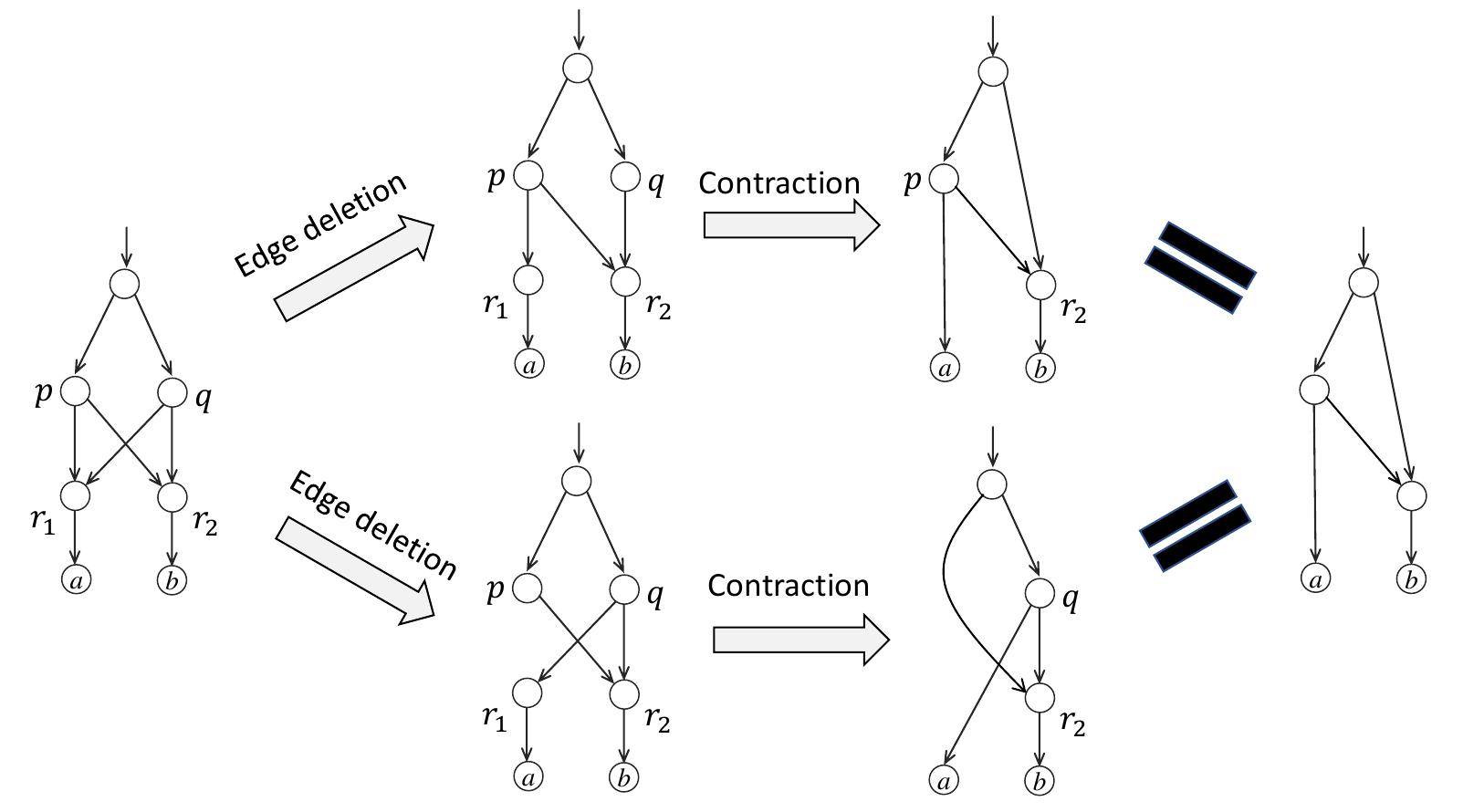}
    \caption{A binary phylogenetic network (BPN) $N$ on taxa $a$ and $b$ (left). 
Deleting $(p, r_1)$ and contracting $p$ and $r_1$ yields the same BPN (right) as deleting $(q, r_1)$ and contracting $q$ and $r_1$. 
Conversely, $N$ can be obtained from the right BPN by a unique edge insertion.
    \label{fig:DeletionC1}
    }
\end{figure}

\begin{figure}[!b]
    \centering
    \includegraphics[width=0.7\textwidth]{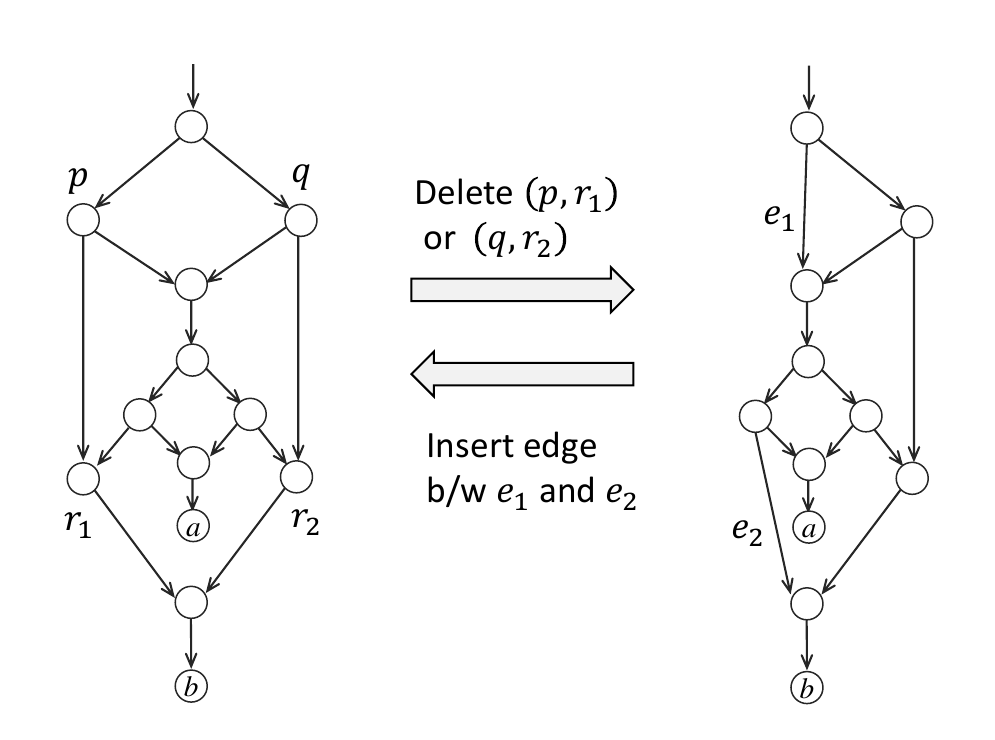}
    \caption{Another binary phylogenetic network (left) that can be obtained by edge insertion in only four distinct ways, instead of eight. Here,
    $r_1, r_2$ and $r_3$ form a reticulation cluster. The edges entering into this cluster from tree nodes are $(p, r_1), (s, r_1), (t, r_2)$, and $(q, r_2)$.
    \label{fig:DeletionC2}
    }
\end{figure}

\newpage
\clearpage

\begin{definition}
Let $N$ be a BPN.

\begin{enumerate}
\item An edge $e=(p,r)\in E(N)$ is called \emph{indispensable} if $p$ is a tree node, 
$r$ is a reticulation node, and either  $p$ or $r$ is contained in a 
3-node (undirected) cycle.

\item An edge $e=(p,r)\in E(N)$ is called \emph{removable} if $p$ is a tree node that is not contained in any 3-node  (undirected)
cycle and $r$ is a reticulation node that is not contained in any 3-node (undirected) cycle.
\end{enumerate}
\end{definition}

It is straightforward to verify the following result.
\begin{proposition}
Let $N$ be a BPN
and $e$ be a directed edge from a tree node to a reticulation node. 
If removing $e$ and contracting its endpoints does not result in a BPN, then
$e$ is indispensable.
\end{proposition}

Let $N$ be a BPN. An \emph{automorphism} of $N$ is a bijection 
$\varphi: V(N) \to V(N)$ such that for all $u,v \in V(N)$,
\[
(u,v) \in E(N) \;\Longleftrightarrow\; (\varphi(u), \varphi(v)) \in E(N),
\]
and  
$\varphi(\ell)=\ell$ for each leaf $\ell$.


\begin{definition}
Let $N$ be a BPN. Let $e_1=(p_1, r_1)$ and $e_2=(p_2, r_2)$ be two distinct edges such that 
$p_1,p_2$ are tree nodes and $r_1,r_2$ are reticulation nodes.
The edges $e_1$ and $e_2$ are said to be \emph{conjugate} if there exists an automorphism $\varphi$ of $N$ such that
$\varphi(e_1)=e_2$.
\end{definition}

Conjugate edges may share an endpoint or have disjoint endpoints. For example, the edges $(p, r_1)$ and $(q, r_1)$
in Figure~\ref{fig:DeletionC1} are conjugate; the edges
$(p, r_1)$ and $(q, r_2)$ in Figure~\ref{fig:DeletionC2} are conjugate.

\begin{proposition}
\label{prop555}
Let $N$ be a  BPN, and let $e_1$ and $e_2$ be edges from tree nodes to reticulation nodes that are removable in $N$.
If the BPN obtained by deleting \(e_1\) and contracting its endpoints is the same as the BPN obtained by deleting \(e_2\) and contracting its endpoints, then \(e_1\) and \(e_2\) are conjugate.
\end{proposition}
\begin{proof}
Let $e_1=(p_1, r_1)$ and $e_2=(p_2, r_2)$, where $p_i$ is a tree node and $r_i$ is a reticulation node for $i=1,2$.
We consider the following cases.

\medskip
\noindent {\bf Case 1.} $p_1=p_2=p$ and $r_1\neq r_2$.

Let $q$ be the (unique) parent of $p$. For $i=1,2$, let $a_i$ be the unique child of $r_i$ and let $b_i$ be the other parent of $r_i$.
Let $N_i$ be the BPN obtained from $N$ by removing $e_i$ and contracting $p$ and $r_i$.

Then
\begin{eqnarray*}
&&V(N_1) = V(N)\setminus \{p,r_1\}, \\
&&E(N_1) = \left[ E(N)\setminus \{(p,r_1)=e_1,(q,p),(b_1,r_1),(r_1,a_1), (p, r_2)\} \right]
          \cup \{(b_1,a_1),(q,r_2)\}, \\
&&V(N_2) = V(N)\setminus \{p,r_2\}, \\
&&E(N_2) = \left[ E(N)\setminus \{(p,r_2)=e_2,(q,p),(b_2,r_2),(r_2,a_2), (p, r_1)\} \right]
          \cup \{(b_2,a_2),(q,r_1)\}.
\end{eqnarray*}

Since inserting an edge between
$(q, r_2)$ and $(b_1, a_1)$ in $N_1$ generates $N$ in the same way as inserting an edge between $(q, r_1)$ and $(b_2, a_2)$ in $N_2$, 
the mapping defined as 
\[
  \phi (x)=\left\{\begin{array}{ll}
        a_2 & \mbox{if } x=a_1 \\
         a_1 & \mbox{if }  x=a_2 \\
            b_2 &\mbox{if }   x=b_1 \\
         b_1 & \mbox{if } x=b_2 \\
         r_1 & \mbox{if } x=r_2 \\
         x &  \mbox{ otherwise}.
    \end{array} \right.
\]
is an isomorphism from $N_1$ to $N_2$.

We extend $\phi$ to a mapping $\varphi: V(N)\to V(N)$ by
\[
\varphi(x)=
\begin{cases}
r_2 & \text{if } x=r_1,\\
p & \text{if } x=p,\\
\phi(x) & \text{otherwise}.
\end{cases}
\]
Then $\varphi$ is an automorphism of $N$, and clearly $\varphi(e_1)=e_2$.

\medskip
\noindent {\bf Case 2.} $p_1\neq p_2$ and $r_1=r_2$.

\medskip
\noindent {\bf Case 3.} $p_1\neq p_2$ and $r_1\neq r_2$.

In both cases, the existence of an automorphism mapping $e_1$ to $e_2$ can be established by arguments analogous to Case~1.
\end{proof}

\begin{figure}[b!]
    \centering
    \includegraphics[width=\textwidth]{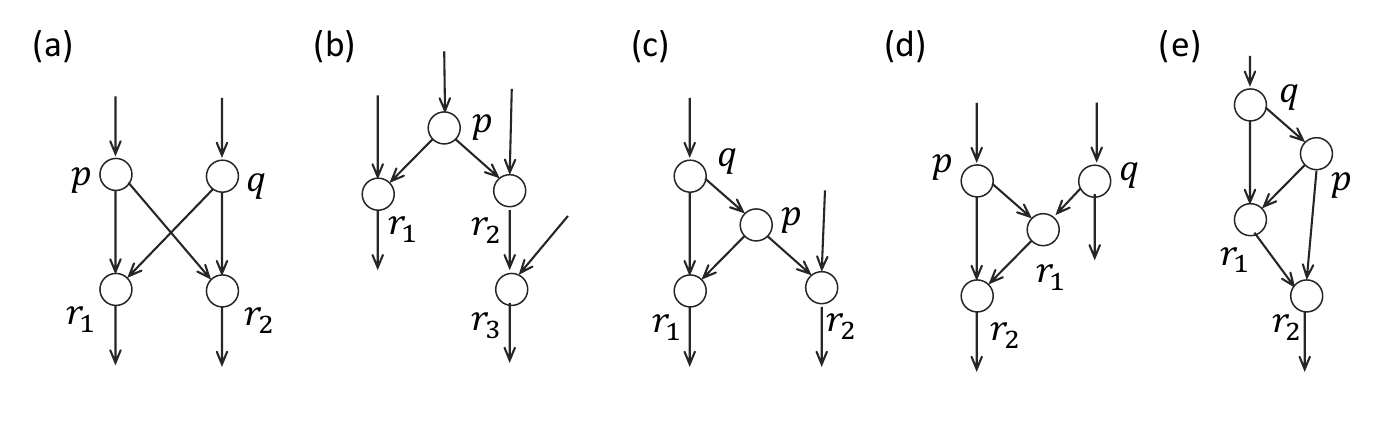}
    \caption{(a) ({\bf Type-1 motif}) Two reticulation nodes $r_1$ and $r_2$ have the same parents.
    (b) ({\bf Type-2 motif}) A tree node $p$ is a common parent of two reticulation nodes $r_1$ and $r_2$. 
The other parents of $r_1$ and $r_2$ are distinct, and $r_2$ is a parent of another reticulation node $r_3$. 
The node $r_3$ may or may not be a child of $r_1$.
    (c) ({\bf Type-3 motif}) A tree node $p$ is a common parent of two reticulation nodes $r_1$ and $r_2$, the parent of $p$ is another parent of $r_1$,  and $r_2$ is not the child of $r_1$.
    (d) ({\bf Type-4 motif}) 
     The children of a tree node $p$ are reticulation nodes $r_1$ and $r_2$, $r_1$ is the other parent of $r_2$, and the other parent $q$ of $r_1$ is not the parent of $p$. 
     (e) ({\bf Type-5 motif})  
      The children of a tree node $p$ are reticulation nodes $r_1$ and $r_2$, $r_1$ is another parent of $r_2$, and 
      the parent $q$ of $p$ is another parent of $r_1$.
    \label{fig:4types}
    }
\end{figure}

We define five types of subgraphs in BPNs as shown in Figure~\ref{fig:4types}.
It is straightforward to verify the following result.
\begin{proposition}
\label{prop4}
Let $N$ be a BPN. Then $N$ contains an indispensable edge if and only if it contains a subgraph of type~3, type~4, or type~5, defined in Panels~(c)--(e) of Figure~\ref{fig:4types}, respectively.
\end{proposition}

\begin{theorem}
\label{theorem111}
Let $N$ be a  BPN with $k$ reticulations that form $c$ clusters, where $c \leq k$. 
If $N$ cannot be obtained in  $k+c$ distinct ways by inserting an edge in a BPN with $k-1$ reticulations, 
then $N$ must contain one of the subgraphs defined in Figure~\ref{fig:4types}.
\end{theorem}
\begin{proof}
Consider a reticulation cluster $C$ in $N$. Assume that $C$ consists of $t$ reticulation nodes. 
Then, by Proposition~\ref{prop3_entering_ret_cluster}, there are $t+1$  edges entering $C$ from tree nodes, that is, edges from tree nodes to  reticulation nodes in $C$ (see Figure~\ref{fig:DeletionC2}). 

Let these edges be
$e_1, e_2, \cdots, e_{t+1}$, where
$e_i=(p_i, r_i)$ for $1\leq i\leq t+1$, 
and let $N_i=N-(p_i, r_i)$. 
If any of these edges is indispensable, then, by Proposition~\ref{prop4},  $N$ contains a subgraph of type-3, type-4, or type-5 defined in Figure~\ref{fig:4types}.

In the rest of the proof, we assume that 
$N$ does not contain any type-3, type-4 or type-5 subgraph.

Since $k$ reticulation nodes form $c$ reticulation clusters, there are 
$k+c$ edges from tree nodes to reticulation nodes that are in these reticulation clusters.
Assume that $N$ cannot be generated from BPNs with $k-1$ reticulations by edge insertion  in $k+c$ ways. There must be two removable edges $e'=(p', r')$ and $e''=(p'', r'')$ such that  $N'=N-e'$ and $N''=N-e''$
are identical, as shown in Figures~\ref{fig:DeletionC1} and \ref{fig:DeletionC2}. Therefore, by Proposition~\ref{prop555},  $e'$ and $e''$ are conjugate. 

Since $N$ is finite, without loss of generality, we may assume that there is no  pair of  conjugate edges $(q', s')$ and $(q'', s'')$ such that $(q', s')$ is below $p'$. 

Let $\varphi$ be an automorphism associated with the conjugation of $e'$ and $e''$, that is,   $\varphi(p')=p''$ and $\varphi(r')=r''$.
Clearly, a leaf below $p'$ if and only if it is below 
$p''$.  
We consider the following two possible cases.

{\bf Case 1}.  $p'=p''$ and $r'\neq r''$. 

Note that  $r'\neq \varphi(r') = r''$.
Let  $\ell$ be a leaf below $r'$. Consider a path $P$ from $r'$ to $\ell$: 
$$P:  r'=v_0, v_1, ..., v_m=\ell, m\geq 1.$$
and its image path $\varphi(P)$  from $r''$ to $\ell$:
$$\varphi(P):  r''=\varphi(v_0), \varphi(v_1), ..., \varphi(v_m)=\ell, m\geq 1.$$
Since each reticulation node has a unique child,  that $r'\neq r''$ implies that 
$v_1\neq \ell $ and $m \geq 2$. 
We consider the smallest index $a$ ($\geq 1$) such that 
$v_{a-1} \neq \varphi(v_{a-1})$ and 
$v_a=\varphi(v_a)$.  This implies that 
$v_i\neq \varphi(v_i)$ for each $i< a$.
Note that $v_a$ is a reticulation node and its parents are $v_{a-1}$ and $\varphi(v_{a-1})$.

If $v_1$ is a reticulation node, then, 
$p'$, $r'=v_0$, $r''$ and $v_1$ induce a subgraph of type-2 in Figure~\ref{fig:4types}. 

If $v_1$ is a tree node, we consider the last tree node $v_b$ in the subpath from $v_1$ to 
$v_a$. 
Because of the choice of $a$,  $v_b\neq \varphi(v_b)$ or 
$v_{b+1}\neq \varphi(v_{b+1})$. This implies that 
$(v_b, v_{b+1})$ and $(\varphi(v_b), \varphi(v_{b+1}))$ are a pair of conjugate edges,
contradicting the hypothesis that there is no pair of conjugate edges below $e'$ and $e''$.

{\bf Case 2.} 
$p'\neq p''$. 
Let 
$\ell$ be a leaf below the other child $x$ of $p'$. Consider a path $P$ from $p'$ to 
$\ell$ :
$$P:  p'=v_0, v_1=x, ..., v_m=\ell, m\geq 1,$$
 and its image path:
$$\varphi(P):  p''=\varphi(v_0), \varphi(v_1), ..., \varphi(v_m)=\ell.$$
Then, there exists $a$ such that 
$v_{i}\neq \varphi(v_{i})$ for each $i< a$
and $v_a=\varphi(v_a)$ in the path. 
Clearly,  
$v_a$ is a reticulation node and its parents are $v_{a-1}$ and $\varphi(v_{a-1})$. 

If $a=1$, then, $p', p'', v_1, r'$ induce a subgraph of type-1 if $r'=r''$. 
If $r'\neq r''$, by 
arguments analogous to Case 1, we can show that the child $s'$ of $r'$ must be a reticulation node and thus
$p', v_1, r'$ and $s'$ induce a subgraph of type-2.

Assume $a>1$. Consider the last tree node $v_t$ in the subpath of $P$ from $v_0$ to  $v_a$, where $0\leq t\leq a-1$. 

If $t=0$,  then $v_i$ is a reticulation node for each $i$ from 1 to $a$. The nodes
\(p'\), \(r'\), \(v_1=x\), and \(v_2\)
induce a subgraph of type-2.


If $t>0$, then, $v_{t+1}$ is a reticulation node and 
$(v_{t}, v_{t+1})$ and $(\varphi(v_{t}), 
\varphi(v_{t+1}))$ are a pair of conjugate edges. This is a contradiction to the hypothesis that there is no edge below $p'$ that forms  a conjugate pair with another edge. 

This completes the proof of the theorem.
\end{proof}

\subsection{Counting BPNs Containing Special Subgraphs}

\begin{proposition}
\label{prop777}
Let ${\cal P}(n, k)$ be the set of BPNs with $k$ reticulations on a set of $n$ taxa.  
For each $i \in \{1,2,3,4,5\}$, let ${\cal C}_i(n, k) \subseteq {\cal P}(n, k)$ be the set of BPNs that contain a node-induced subgraph of type-$i$ defined in Figure~\ref{fig:4types}. Then
\begin{enumerate}
\item $|{\cal C}_1(n, k)| \leq 2(k-1) \, |{\cal P}(n, k-1)|$.
\item $|{\cal C}_2(n, k)| \leq (k-1)(4k-6) \, |{\cal P}(n, k-1)|$.
\item $|{\cal C}_3(n, k)| \leq 2(k-1) \, |{\cal P}(n, k-1)|$.
\item $|{\cal C}_4(n, k)| \leq 2(k-1) \, |{\cal P}(n, k-1)|$.
\item $|{\cal C}_5(n, k)| \leq 2k \, |{\cal P}(n, k-1)|$.
\end{enumerate}
\end{proposition}
\begin{proof}
(1)
Let $N \in {\cal C}_1(n,k)$, and suppose that $p', p'', r', r''$ are four nodes that induce a subgraph of type-1 in $N$, where $p'$ and $p''$ are tree nodes and $r'$ and $r''$ are reticulation nodes.  
Let $q''$ be the parent of $p''$, and let $s'$ be the child of $r'$.

Delete the edge $(p'', r')$ and then contract the nodes $p''$ and $r'$. 
The resulting BPN $N'$ is a BPN in ${\cal P}(n,k-1)$. 
In $N'$, the node $r''$ remains a reticulation node with parents $p'$ and $q''$, and $s'$ is the sibling of $r''$.
Conversely, inserting an edge from $(q'', r'')$ to $(p', s')$ in $N'$ reconstructs $N$. 


Based on this observation, all BPNs in ${\cal C}_1(n,k)$ can be generated from
BPNs of ${\cal P}(n,k-1)$  by the following edge insertion procedure:
\newpage

\begin{center}
\rule{13.5cm}{0.5pt}
\begin{algorithmic}
\For{each BPN $B \in {\cal P}(n,k-1)$}
    \For{each reticulation node $r$ in $B$ with parents $u$ and $v$}
        \If{$u$ is a tree node with the other child $x$}
            \State Construct a BPN by inserting an edge from $(v,r)$ to $(u,x)$.
        \EndIf
        \If{$v$ is a tree node with the other child $y$}
            \State Construct a BPN by inserting an edge from $(u,r)$ to $(v, y)$.
        \EndIf
    \EndFor
\EndFor
\end{algorithmic}
\rule{13.5cm}{0.5pt}
\end{center}

Since each BPN of ${\cal P}(n,k-1)$ contains $k-1$ reticulation nodes and each reticulation node has at most two parents that are tree nodes, we obtain the upper bound on $\vert {\cal C}_1(n,k)\vert$.
\\

(2) Let $N \in {\cal C}_2(n,k)$, and suppose that $p, r_1, r_2, r_3$ are four nodes of $N$ that induce a subgraph of type-2, as shown in Figure~\ref{fig:4types}b. 
We say that two reticulation nodes form a \emph{stacked reticulation pair} if one is the child of the other, as in the case of $r_2$ and $r_3$.

First, consider the special case in which $r_3$ is also the unique child of $r_1$. 
In this case, deleting the edge $(p,r_1)$ and contracting $p$ and $r_1$ yields a BPN $N' \in {\cal P}(n,k-1)$ in which $r_2$ and $r_3$ form a stacked reticulation pair. Conversely, $N$ can be reconstructed from $N'$ by inserting an edge from an incoming edge of $r_2$ to the incoming edge of $r_3$ that is not from $r_2$. 
Therefore, the number of BPNs containing a type-2 induced subgraph with this special property is at most
\[
2 \times (k-2) \times \vert {\cal P}(n,k-1)\vert 
= 2(k-2)\vert {\cal P}(n,k-1)\vert.
\]
This is because $k-1$ reticulation nodes can form at most $k-2$ stacked pairs.

Next, suppose that $r_3$ is not a child of $r_1$. 
Deleting the edge $(p,r_2)$ and contracting $p$ and $r_2$ yields a BPN $N' \in {\cal P}(n,k-1)$ in which $r_1$ and $r_3$ remain two unrelated reticulation nodes. Conversely, $N$ can be reconstructed from $N'$ by inserting an edge between an incoming edge of $r_1$ and an incoming edge of $r_3$. 
Therefore, the number of BPNs containing a type-2 induced subgraph in which the bottom reticulation in the stacked pair is not a child of the other reticulation node is at most
\[
2\times 4 \times {k-1 \choose 2} \times \vert {\cal P}(n,k-1)\vert 
= 4(k-1)(k-2)\vert {\cal P}(n,k-1)\vert.
\]

Combining the above two bounds yields the desired upper bound.
\\

(3) Let $N \in {\cal C}_3(n,k)$, and suppose that $p, q, r_1, r_2$ are four nodes of $N$ that induce a subgraph of type-3, as shown in 
Figure~\ref{fig:4types}c. 
Let  $s_1$ be the child of $r_1$.  

Delete the edge $(p, r_1)$ and then contract the nodes $p$ and $r_1$. 
The resulting BPN $N'$ is a BPN in ${\cal P}(n,k-1)$. In $N'$, the node $r_2$ remains a reticulation node, $q$ replaces $p$ as a parent of $r_2$ and replaces $r_1$ as  a parent of $s_1$. 
Conversely, inserting an edge from 
$(q, r_2)$ to $(q, s_1)$ in $N'$ reconstructs $N$.

Based on this observation, all BPNs in ${\cal C}_3(n,k)$ can be generated from
BPNs of ${\cal P}(n,k-1)$  by the following edge insertion procedure:
\\

\begin{center}
\rule{13.5cm}{0.5pt}
\begin{algorithmic}
\For{each BPN $B \in {\cal P}(n,k-1)$}
    \For{each reticulation node $r$ in $B$ with parents $u$ and $v$}
        \If{$u$ is a tree node with the other child $x$}
            \State Construct a BPN by inserting an edge from $(u,r)$ to $(u,x)$.
        \EndIf
        \If{$v$ is a tree node with the other child $y$}
            \State Construct a BPN by inserting an edge from $(v,r)$ to $(v, y)$.
        \EndIf
    \EndFor
\EndFor
\end{algorithmic}
\rule{13.5cm}{0.5pt}
\end{center}

Since each BPN of ${\cal P}(n,k-1)$ contains $k-1$ reticulation nodes and each reticulation node has at most two parents that are tree nodes, $\vert {\cal C}_3(n,k)\vert
\leq 2(k-1) \, |{\cal P}(n, k-1)|$.
\\

(4) By arguments analogous to Parts (1) and (3), all BPNs in ${\cal C}_4(n,k)$ can be generated from BPNs of ${\cal P}(n,k-1)$  using  the following edge insertion procedure:
\begin{center}
\rule{13.5cm}{0.5pt}
\begin{algorithmic}
\For{each BPN $B \in {\cal P}(n,k-1)$}
    \For{each reticulation node $r$ in $B$ with parents $u$ and $v$}
        \If{$v$ is a tree node}
            \State Construct a BPN by inserting an edge from $(u,r)$ to $(v, r)$.
        \EndIf
         \If{$u$ is a tree node}
            \State Construct a BPN by inserting an edge from $(v,r)$ to $(u, r)$.
        \EndIf
    \EndFor
\EndFor
\end{algorithmic}
\rule{13.5cm}{0.5pt}
\end{center}

Since there are $k-1$ reticulation nodes, this edge insertion procedure leads to an upper bound for counting BPNs in ${\cal C}_4(n,k)$.
\\

(5) Let $N \in {\cal P}(n,k)$ be a BPN that contains a subgraph $G$ of type-5, induced by $\{p, q, r_1, r_2\}$ as shown in Figure~\ref{fig:4types}.

First, suppose that the top node $q$ of $G$ is a child of a tree node $u$, as illustrated in the top-right panel of Figure~\ref{t5_transfer1}. By relocating $q$ to another outgoing edge of $u$ and deleting the edge from $p$ to $r_1$, we obtain a BPN $N' \in {\cal P}(n,k-1)$ in which $u$, $q$, and $r_2$ induce a triangle. Conversely, $N$ can be recovered uniquely from $N'$ by reversing this procedure.

Next, suppose that the top node $q$ of $G$ is a child of a reticulation node $s$. Applying the same transformation yields a BPN $N' \in {\cal P}(n,k-1)$ in which $s$, $q$, and $r_2$ induce a triangle. Again, the reverse operation reconstructs $N$ from $N'$.

Therefore, by Part (4),  the number of BPNs in which a type-5 induced subgraph has its top node below an internal node is at most $2(k-1)\,\vert {\cal P}(n,k-1) \vert$. 

Moreover, BPNs in which the top node of a type-5 induced subgraph is the child of the root are in one-to-one correspondence with BPNs in ${\cal P}(n,k-2)$.

In total, the number of BPNs containing a type-5 induced subgraph is at most
\[
2(k-1)\,\vert {\cal P}(n,k-1) \vert
+ \vert {\cal P}(n,k-2) \vert
\leq 2k\,\vert {\cal P}(n,k-1) \vert.
\]
The proof is done.
\end{proof}

\begin{figure}[htbp]
    \centering
    \includegraphics[width=0.4\textwidth]{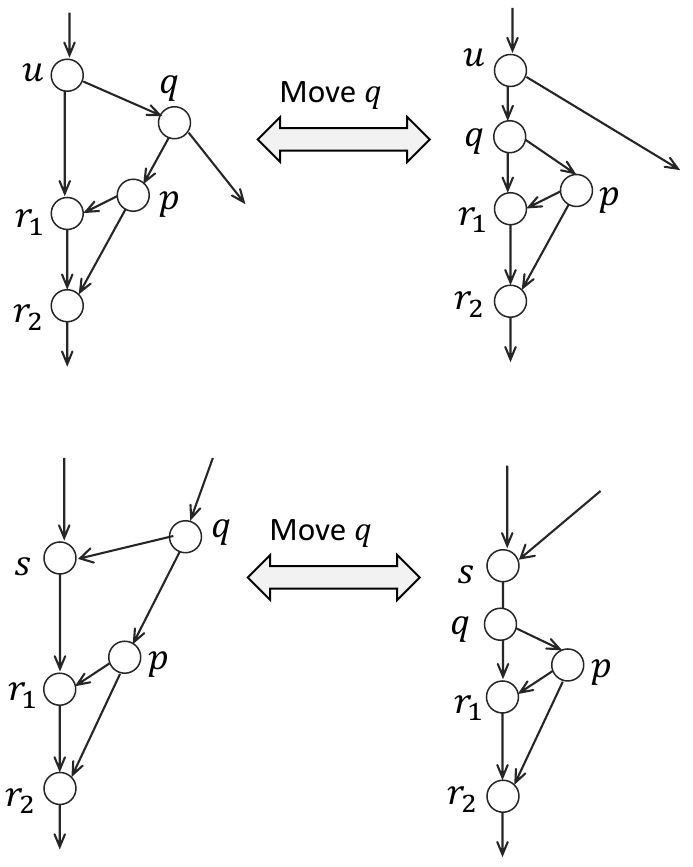}
    \caption{
    Two methods for constructing BPNs that contain an induced subgraph of type-5 via node relocation. 
{\bf Top}: A construction in which the top node $q$ of the induced type-5 subgraph is a child of a tree node $u$. 
{\bf Bottom}: A construction in which the top node $q$ is a child of a reticulation node $s$. 
    \label{t5_transfer1}}
\end{figure}


\section{The Main Theorem}

By Proposition~\ref{prop777}, we obtain the following inequality:
\begin{eqnarray}
 && \vert \cup^{5}_{i=1} {\cal C}_i(n, k)\vert  \nonumber \\
 &\leq &
 \{6(k-1)+(k-1)(4k-6)+2k\} \cdot \vert {\cal P}(n, k-1)\vert \nonumber  \\
 &=& 2(2k^2-k) \vert {\cal P}(n, k-1)\vert  \label{eqn33333}
 \end{eqnarray}

Let ${\cal Q}(n, k)={\cal P}(n, k)\setminus \left(\cup^{5}_{i=1} {\cal C}_i(n, k)\right)$.
Let ${\cal Q}(n,k,c)$ denote the set of BPNs in ${\cal Q}(n,k)$ with exactly $c$ reticulation clusters.
By Theorem~\ref{theorem111}, each BPN in ${\cal Q}(n, k, c)$ can be obtained from BPNs in ${\cal P}(n, k-1)$ by 
inserting an edge in $k+c$ different ways. 
Since every BPN in ${\cal P}(n,k-1)$ has $2n+3k-4$ edges, 
\begin{eqnarray}
  \sum^{k}_{c=1}(k+c) \vert {\cal Q}(n, k, c) \vert 
  \leq (2n+3k-4)(2n+3k-5)\vert {\cal P}(n, k-1) \vert.
  \label{eqn222}
\end{eqnarray}

In addition, 
suppose that $N\in {\cal Q}(n,k,c)$  contains 
$c'$ reticulation clusters with sizes $t_1,\ldots,t_{c'}$ and 
$c-c'$ reticulation clusters of size 1, where 
$t_i>1$ for each $i$. Then we can obtain $N$ in 
$\sum_{i=1}^{c'}t_i$ different ways. Note that
$\sum_{i=1}^{c'}t_i\geq (c'-c)(1-1)+\sum_{i=1}^{c'}(t_i-1) \geq k-c$. 
Since each BPN in ${\cal P}(n,k-1)$ has at most $2(k-1)$ edges entering 
reticulation clusters, we obtain
\begin{eqnarray}
  \sum^{k}_{c=1}(k-c) \vert {\cal Q}(n, k, c) \vert 
  \leq 2(k-1)(2n+3k-4)\vert {\cal P}(n, k-1) \vert.
  \label{eqn333}
\end{eqnarray}

Combining Inequalities (\ref{eqn222}) and 
(\ref{eqn333}), we obtain the following inequalities:
\begin{eqnarray*}
& & 2k \vert {\cal Q}(n, k) \vert  \nonumber \\
&= & \sum^{k}_{c=1}(k+c) \vert {\cal Q}(n, k, c) \vert 
 + \sum^{k}_{c=1}(k-c) \vert {\cal Q}(n, k, c) \vert  \nonumber \\
&\leq & [(2n+3k-4)(2n+3k-5)+2(2n+3k-4) (k-1)]\;\vert {\cal P}(n, k-1) \vert \\
&= & (2n+3k-4) (2n+5k-7) \vert {\cal P}(n, k-1) \vert.
\end{eqnarray*}
and, by (\ref{eqn33333}), 
\begin{eqnarray}
  && 2k \vert{\cal P}(n, k)\vert \nonumber \\
  &\leq & 2k \vert {\cal Q}(n, k) \vert +  2k \vert \cup^{5}_{i=1} {\cal C}_i(n, k)\vert \nonumber\\
  &\leq & (2n+3k-4) (2n+5k-7) \vert {\cal P}(n, k-1) \vert
  + 2k(4k^2-2k) \vert {\cal P}(n, k-1) \vert  \nonumber \\
  &\leq & (2n+3k-4) (2n+9k-2) \vert {\cal P}(n, k-1) \vert  \label{inequality555}
\end{eqnarray}
for any $1\leq k < \sqrt{n}$.

\begin{theorem}
\label{thm2}
   Let ${\cal P}(n, k)$ denote the set of BPNs with $k$ reticulation nodes on a set of $n$ taxa. 
   For any $k$ such that $1\leq k < \sqrt{n}$, 
    \begin{equation}\label{eqn322}
        |{\cal P}(n,k)| \leq {n\choose k} \frac{2^{n+k-1/2}n^{n+k-1}}{e^n}e^{\frac{7k^2}{n} 
    +\frac{1}{12n}+\frac{1}{12(n-k)}}.
    \end{equation} 
\end{theorem}
\begin{proof}
  For $n\ge 2$, by applying Equation~(\ref{eqn22222}) and Inequality (\ref{inequality555}) repeatedly, we obtain:

    \begin{eqnarray*}
    \vert {\cal P}(n,k) \vert &\le& \prod_{j=1}^k \left[\frac{(2n+3j-4)(2n+9j-2)}{2j}\right] \vert {\cal P}(n,0)\vert    \\
    &=& \frac{\prod^k_{j=1}(2n+3j-4)(2n+9j-2)}{2^k k!}\frac{(2n-2)!}{2^{n-1}(n-1)!}   \\
    &\leq & \frac{\prod^k_{j=1}(2n+6j-3)^2}{2^k k!}\frac{(2n-2)!}{2^{n-1}(n-1)!}   \\
    &\leq & \frac{\prod^k_{j=1}(2n+6j-3)(2n+6j-2)}{2^k k!}\frac{(2n-2)!}{2^{n-1}(n-1)!}   \\
    &\leq & \frac{\prod^{2k-1}_{j=0}(2n+6k-2-j)}{2^k k!}\frac{(2n-2)!}{2^{n-1}(n-1)!}   \\
    &=& {n\choose k}  \frac{(2n+6k-2)!(2n)!(n-k)!}{2^{n+k}(2n-1)(2n+4k-2)!n!n!}  
    \end{eqnarray*}
Using  
    \begin{eqnarray}
\sqrt{2\pi m}\left(\frac{m}{e}\right)^m
< m!
<
\sqrt{2\pi m}\left(\frac{m}{e}\right)^m e^{\frac{1}{12m}} \label{factorial_to_power}
\end{eqnarray}
for any natural number $m$, we further have:  
    \begin{eqnarray*}
   \vert {\cal P}(n,k) \vert  &\leq & \binom{n}{k} \frac{(2n+6k-2)^{2n+6k-2+1/2}(2n)^{2n+1/2}(n-k)^{n-k+1/2}e^{(2n+4k-2)+n+n}}{2^{n+k}(2n-1)(2n+4k-2)^{2n+4k-2+1/2}n^{2n+1}e^{(2n+6k-2)+2n+(n-k)}}\\
    && \times e^{\frac{1}{24(n+3k-1)}+\frac{1}{24n}+
    \frac{1}{12(n-k)}}\\
    &\leq & \binom{n}{k} \frac{2^{n+k+1/2}}{(2n-1)n^{1/2}e^{n+k}}\frac{(n+3k-1)^{2n+6k-2+1/2}(n-k)^{n-k+1/2}}{(n+2k-1)^{2n+4k-2+1/2}}\\ 
    && \times e^{\frac{1}{12n}+
    \frac{1}{12(n-k)}} \\
    &=& \binom{n}{k} \frac{2^{n+k+1/2}n^{n+k}}{(2n-1)e^{n+k}} \left(1+ \frac{3k-1}{n} \right)^{2k} \left(1+\frac{k}{n+2k-1}\right)^{2n+4k-3/2} \left(1-\frac{k}{n}\right)^{n-k+1/2}\\
    && \times e^{\frac{1}{12n}+\frac{1}{12(n-k)}}\\
    &\leq& \binom{n}{k} \frac{2^{n+k+1/2}n^{n+k}}{(2n-1)e^{n+k}} e^{\frac{2k(3k-1)}{n}}e^{\frac{k(2n+4k-3/2)}{n+2k-1}} e^{-\frac{k(n-k+1/2)}{n}}\\
    && \times e^{\frac{1}{12n}+ \frac{1}{12(n-k)}}\\
    &=& \binom{n}{k} \frac{2^{n+k+1/2}n^{n+k}}{(2n-1)e^{n}} e^{\frac{2k(3k-1)}{n}}e^{\frac{k}{2(n+2k-1)}} e^{\frac{k(2k-1)}{2n}}\\
    && \times e^{\frac{1}{12n}+\frac{1}{12(n-k)}}\\
    &\leq & {n\choose k} \frac{2^{n+k-1/2}n^{n+k-1}}{e^n} e^{\frac{7k^2}{n} 
    +\frac{1}{12n}+\frac{1}{12(n-k)}}
    \end{eqnarray*}
where 
the second inequality follows by combining common terms and replacing 
    $e^{1/[24(n+3k-1)]}$ with $e^{1/(24n)}$, 
    the 3rd inequality using 
\begin{eqnarray}
1+x\le e^{x} 
\label{bounds_power}
\end{eqnarray}
for any real $x$ such that $1\geq x \geq -1$, and 
the last inequality by replacing both $e^{k/(2(n+2k-1))}$ with $e^{k/(2n)}$ and $e^{2k(3k-1)/n}$ with $e^{6k^2/n}$.
\end{proof}

Let ${\cal TC}(n,k)$ denote the set of tree-child networks with $k$ reticulations on $n$ taxa.  
In their study of the enumeration of tree-child networks, Pons and Batle posed a conjecture stating that the number of tree-child networks with $k$ reticulations on $n$ taxa satisfies the following recurrence relation \cite{Conjecture_Pons}:
\[
\vert {\cal TC}(n,k)\vert
=
(n+1-k)\,\vert {\cal TC}(n,k-1)\vert 
+
\frac{n(2n+k-3)}{n-k}\,\vert {\cal TC}(n-1,k)\vert,
\]
where $n>k\geq 1$.
Recently, Lin et al. proved an equivalent form of this conjecture using Young tableaux \cite{chang2024,Conjecture_Proof}.

This recurrence relation immediately implies the following asymptotic formula for the number of tree-child networks; its proof is given in Appendix~A. 
The formula was previously proved for fixed $k$ in \cite{Conjecture_Pons}.

\begin{proposition}
\label{prop10}
For any $n$ and $k$ such that $1\leq k=o(\sqrt{n})$,
\[
|{\cal TC}(n, k)| =
\binom{n}{k}\frac{2^{n+k-1/2}\,n^{n-1+k}}{e^{n}}
\left\{
1-\frac{\sqrt{\pi}}{2}\left(\frac{k}{n^{1/2}}\right)
+{O\left(\frac{k^2}{n}
\right)}
\right\}.
\]
\end{proposition}

Combining the upper bound in Theorem~\ref{thm2} and 
Proposition~\ref{prop10}, we obtain the following result from the fact that 
${\cal TC}(n, k) \subseteq {\cal P}(n, k)$.

\begin{theorem}
\label{thm3}
Let ${\cal P}(n,k)$ denote the set of BPNs with $k$ reticulations on $n$ taxa. 
If $1\leq k=o(\sqrt{n})$, then
\[
\vert {\cal P}(n,k)\vert
\sim
\binom{n}{k}\frac{2^{n+k-1/2}\,n^{n-1+k}}{e^{n}} \sim \vert{\cal TC}(n, k)\vert.
\]
\end{theorem}
Note that, by Proposition~\ref{prop10}, the range of $k$  in Theorem~\ref{thm3} is tight.

\begin{figure}[htbp]
    \centering
    \includegraphics[width=0.3\textwidth]{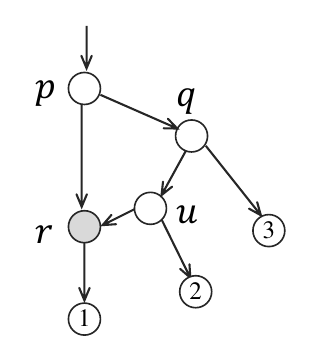}
    \caption{Illustration of free tree nodes and tree edges in tree-child networks. The displayed tree-child network has only one free tree node $q$; the two free edges are $(q, u)$ and $(q, 3)$. 
    \label{free_edge}
    }
\end{figure}

\section{Asymptotic Analysis for Normal Networks}

In this section, we prove the same asymptotic result for normal networks by comparing the number of tree-child networks with the number of normal networks on the same taxon set. 
Throughout this section, we use ${\cal NN}(n,k)$ to denote the set of all normal networks with $k$ reticulations on a fixed set of $n$ taxa, and ${\cal TC}(n,k)$ to denote the set of all tree-child networks with $k$ reticulations on the same set of taxa.

Let $N$ be a tree-child network. 
A tree node is said to be {\it free} if it has out-degree 2 and both of its children are tree nodes or leaves \cite{fuchs2021asymptotic}. 
The two edges from a free tree node to its children are called {\it free tree edges}, as shown in Figure~\ref{free_edge}. 
For the rest of this section, ${\rm FE}(N)$ denotes the set of free tree edges in $N$.

\begin{proposition}
\label{prop11}
In every tree-child network $N\in {\cal TC}(n,k-1)$, there exist (i) $2(n-k)$ free tree edges,  (ii) 
$3k-2$ non-free tree edges and (iii) $(n-k)$ free tree nodes.
\end{proposition}
\begin{proof}
In every tree-child network $N\in {\cal TC}(n,k-1)$, there are 
$2n+k-2$ tree edges and $n+k-2$ tree nodes. 

By the tree-child property, each parent of a reticulation node is a tree node, and its other child is also a tree node or a leaf. However, such a parent is not a free tree node, because one of its children is a reticulation node. Hence the $2(k-1)$ tree edges from these parents to the siblings of the reticulation nodes are not free.

Moreover, the $k-1$ tree edges leaving the reticulation nodes are not free, since their tail endpoints are reticulation nodes rather than tree nodes. The unique edge leaving the root is also not free. Altogether, this accounts for
\[
2(k-1)+(k-1)+1=3k-2
\]
non-free tree edges. Therefore, the number of free tree edges is
\[
(2n+k-2)-(3k-2)=2(n-k).
\]
Since exactly two free tree edges leave each free tree node, the number of free tree nodes is \(n-k\).
\end{proof}

The following fact is straightforward.

\begin{proposition}
\label{prop12}
Let $N$ be a tree-child network in ${\cal TC}(n,k-1)$, and let 
$e'=(u,v)$ and $e''=(p,q)$ be two distinct tree edges in $N$.
Let $M$ be the network generated from $N$ by inserting an edge from $e'$ to $e''$. Then the following statements hold.

\noindent{\rm (1)} The network $M$ is tree-child if and only if the following conditions hold:
\begin{itemize}
    \item $e''$ is a free tree edge;
    \item $e'$ does not lie below the node $q$.
\end{itemize}

\noindent{\rm (2)} Assume that $N$ is normal. The network $M$ is normal if and only if the following conditions hold:
\begin{itemize}
    \item $e''$ is a free tree edge;
    \item $e'$ does not lie below the node $p$, which also implies that $e'\neq e''$;
    \item $e'$ does not lie above the node $p$.
\end{itemize}
\end{proposition}

\begin{proposition}  
\label{prop13}
Let ${\rm low}_N(e)$ denote the number of tree edges below $q$ for a free tree edge $e=(p,q)$ in a tree-child network $N$. Then, for $1\leq k<n$,
\begin{eqnarray}
\label{eqn11}
\vert {\cal TC}(n,k)\vert=
\frac{2(2n+k-3)(n-k)\vert {\cal TC}(n,k-1)\vert
-\displaystyle\sum_{N\in {\cal TC}(n,k-1)}
\sum_{e\in {\rm FE}(N)}{\rm low}_N(e)
}{2k}.
\end{eqnarray}
\end{proposition}

\begin{proof}
By Proposition~\ref{prop11}, each tree-child network 
$N\in {\cal TC}(n,k-1)$ contains $2n+k-2$ tree edges, of which $2(n-k)$ are free tree edges. 

Let $e=(p,q)$ be a free tree edge of $N$. For each tree edge $(u,v)$ that is distinct from $e$ and is not below $q$ in $N$, Proposition~\ref{prop12} implies that inserting an edge from $(u,v)$ to $e$ produces a tree-child network in ${\cal TC}(n,k)$. 
For this fixed free tree edge $e$, there are
\[
2n+k-3-{\rm low}_N(e)
\]
such choices of $(u,v)$. Therefore, the number of tree-child networks generated from $N$ by inserting an edge into a free tree edge is
\[
\sum_{e\in {\rm FE}(N)}
\left(2n+k-3-{\rm low}_N(e)\right)=
2(2n+k-3)(n-k)-\sum_{e\in {\rm FE}(N)}{\rm low}_N(e).
\]

Summing over all networks $N\in {\cal TC}(n,k-1)$ gives
\[
2(2n+k-3)(n-k)\vert {\cal TC}(n,k-1)\vert
-\sum_{N\in {\cal TC}(n,k-1)}
\sum_{e\in {\rm FE}(N)}{\rm low}_N(e)
\]
edge-insertion constructions that produce tree-child networks in ${\cal TC}(n,k)$.

However, the same network in ${\cal TC}(n,k)$ may be obtained from networks in ${\cal TC}(n,k-1)$ in more than one way. More precisely, each tree-child network in ${\cal TC}(n,k)$ can be generated from networks in ${\cal TC}(n,k-1)$ by edge insertion in exactly $2k$ distinct ways. Dividing by $2k$ gives Equality~\eqref{eqn11}.
\end{proof}

Using an argument analogous to the proof of Proposition~\ref{prop13}, we obtain the following identity for normal networks.

\begin{proposition}
\label{prop14}
Let ${\rm top}_{N}(e)$ denote the number of tree edges that lie above $p$ for a free tree edge $e=(p,q)$ in a tree-child network $N$, and let ${\rm low}_N(e)$ be defined as in Proposition~\ref{prop13}. 
Let ${\cal NN}(n,k)$ denote the set of normal networks with $k$ reticulations on a fixed set of $n$ taxa. 
Then
\begin{equation}
\label{eqn13}
\begin{aligned}
\vert {\cal NN}(n,k)\vert
={}&
\frac{(2n+k-4)(2n-2k)\vert {\cal NN}(n,k-1)\vert}{2k} \\
&-
\frac{
\displaystyle\sum_{N\in {\cal NN}(n,k-1)}
\sum_{e\in {\rm FE}(N)}
\left[
{\rm top}_N(e)+{\rm low}_N(e)+{\rm low}_N(f(e))
\right]
}{2k},
\end{aligned}
\end{equation}
where \(f(e)\) denotes the other free tree edge outgoing from the same start endpoint \(p\) of \(e\).
\end{proposition}

\begin{proposition}
\label{prop15}
Let $N\in {\cal NN}(n,k-1)$. Then the following inequality and equality hold:
\begin{eqnarray}
   \sum_{e\in {\rm FE}(N)}{\rm top}_{N}(e) 
   \leq 
   \sum_{e\in {\rm FE}(N)}{\rm low}_{N}(e) 
   + 2(3k-2)(n-k),
    \label{eqn14}\\
   \sum_{e\in {\rm FE}(N)}
   \left[{\rm low}_N(e)+{\rm low}_N(f(e))\right]
   =
   2\sum_{e\in {\rm FE}(N)}{\rm low}_N(e),
   \label{eqn15}
\end{eqnarray}
where $f(e)$ denotes the other free tree edge with the same start endpoint as $e$.
\end{proposition}

\begin{proof}
Let $N\in {\cal NN}(n,k-1)$. 
Recall that ${\rm TE}(N)$ and ${\rm FE}(N)$ denote the sets of tree edges and free tree edges of $N$, respectively. 
Let ${\rm NE}(N)$ denote the set of non-free tree edges. 
By Proposition~\ref{prop11}, we have 
\[
|{\rm FE}(N)|=2(n-k)
\qquad\mbox{and}\qquad
|{\rm NE}(N)|=3k-2.
\]

For two tree edges $e_1=(u,v)$ and $e_2$, write $e_1\twoheadrightarrow e_2$ if $e_2$ lies below the endpoint $v$ of $e_1$. 
Since ${\rm TE}(N)$ is the disjoint union of ${\rm FE}(N)$ and ${\rm NE}(N)$, we have
\begin{align*}
\sum_{e\in {\rm FE}(N)}{\rm top}_{N}(e) 
&=
\sum_{e\in {\rm FE}(N)}
\left|\{e'\in {\rm TE}(N): e'\twoheadrightarrow e\}\right| \\
&=
\sum_{e\in {\rm FE}(N)}
\left|\{e'\in {\rm FE}(N): e'\twoheadrightarrow e\}\right| \\
&\quad+
\sum_{e\in {\rm FE}(N)}
\left|\{e'\in {\rm NE}(N): e'\twoheadrightarrow e\}\right| \\
&\leq
\sum_{e\in {\rm FE}(N)}
\left|\{e'\in {\rm FE}(N): e'\twoheadrightarrow e\}\right|
+
|{\rm FE}(N)|\,|{\rm NE}(N)| \\
&=
\sum_{e\in {\rm FE}(N)}
\left|\{e'\in {\rm FE}(N): e'\twoheadrightarrow e\}\right|
+
2(n-k)(3k-2) \\
&\leq
\sum_{e\in {\rm FE}(N)}{\rm low}_{N}(e)
+
2(n-k)(3k-2).
\end{align*}
This proves Inequality~\eqref{eqn14}.

For every free tree edge $e$, the edge $f(e)$ is the other free tree edge with the same tail as $e$. 
Thus the map $f:{\rm FE}(N)\to {\rm FE}(N)$ is an involution without fixed points; that is,
\[
f(f(e))=e
\qquad\mbox{and}\qquad
f(e)\neq e.
\]
Hence
\[
\sum_{e\in {\rm FE}(N)}{\rm low}_N(f(e))
=
\sum_{e\in {\rm FE}(N)}{\rm low}_N(e).
\]
Therefore,
\[
\sum_{e\in {\rm FE}(N)}
\left[{\rm low}_N(e)+{\rm low}_N(f(e))\right]
=
2\sum_{e\in {\rm FE}(N)}{\rm low}_N(e),
\]
which proves Equality~\eqref{eqn15}.
\end{proof}

Applying Inequalities~(\ref{eqn14}) and (\ref{eqn15}) to Equation~(\ref{eqn13}) and using Equation~(\ref{eqn11}),
we obtain:
\begin{eqnarray*}
\vert {\cal NN}({n,k})\vert
&\geq &\frac{2(2n+k-4)(n-k)\vert {\cal NN}(n,k-1)\vert}{2k} \nonumber \\
&& -\frac{2(3k-2)(n-k)\vert {\cal NN}(n,k-1)\vert+3\sum_{N\in {\cal NN}(n, k-1)}\sum_{e\in {\rm FE}(N)}{\rm low}_N(e)}{2k}\\
&=& \frac{(2n-2k-2)(n-k)\vert {\cal NN}(n,k-1)\vert}{k} - \frac{3\sum_{N\in {\cal NN}(n, k-1)}\sum_{e\in {\rm FE}(N)}{\rm low}_N(e)}{2k}\\
&\geq& \frac{(2n-2k-2)(n-k)\vert {\cal NN}(n,k-1)\vert}{k} - \frac{3\sum_{N\in {\cal TC}(n, k-1)}\sum_{e\in {\rm FE}(N)}{\rm low}_N(e)}{2k}.\\
&=& \frac{(2n-2k-2)(n-k)\vert {\cal NN}(n,k-1)\vert}{k} \\
 && - 3\left(\frac{(2n+k-3)(n-k)\vert {\cal TC}(n,k-1)\vert}{k} -\vert {\cal TC}({n,k})\vert \right).
\end{eqnarray*}
Dividing by ${\cal TC}(n, k)$ on both sides of the above inequality, we have:
\begin{eqnarray*}
\frac{\vert {\cal NN}({n,k})\vert}{\vert {\cal TC}({n,k})\vert} &\geq& 
\frac{(2n-2k-2)(n-k)\vert {\cal NN}(n,k-1)\vert}{k\vert {\cal TC}({n,k})\vert} -\frac{3(2n+k-3)(n-k)\vert {\cal TC}(n,k-1)\vert}{k\;\vert{\cal TC}({n,k})\vert}+3\\
&=& \frac{(2n-2k-2)(n-k)}{k}\cdot \frac{\vert {\cal NN}(n,k-1)\vert}{\vert {\cal TC}({n,k-1})\vert} \cdot \frac{\vert {\cal TC}(n,k-1)\vert}{\vert {\cal TC}({n,k})\vert}
\\ 
&& -\frac{3(2n+k-3)(n-k)}{k}\cdot \frac{\vert {\cal TC}(n,k-1)\vert}{\vert{\cal TC}({n,k})\vert}+3\\
\end{eqnarray*}
Proposition~\ref{prop10} implies the following equality:
\[
\frac{|{\cal TC}(n,k-1)|}{|{\cal TC}(n,k)|}
=
\frac{k}{2n(n-k+1)}
\left[
1+\frac{\sqrt{\pi}}{2\sqrt n}
+o\left(\frac1{\sqrt n}\right)
\right].
\]
for $1\leq k=o(\sqrt{n})$ (Proposition~B1 in the appendix). 
Therefore, using 
 \[\frac{(n-k-1)(n-k)}{n(n-k+1)}
=\left(1-\frac{k}{n}\right)\left(1-\frac{2}{n-k+1}\right)\ge 1-\frac{k+2}{n-k+1} \]
and 
{\[\frac{(2n+k-3)(n-k)}{2n(n-k+1)}
= \left(1+\frac{k-3}{2n}\right)\left(1-\frac{1}{n-k+1}\right)\le1+\frac{k-3}{2n}\]
for $1\le k<n$ , }
we further have: 
\begin{eqnarray*}
\frac{\vert {\cal NN}({n,k})\vert}{\vert {\cal TC}({n,k})\vert} 
&\geq & \frac{(n-k-1)(n-k)}{n(n-k+1)}\cdot \frac{\vert {\cal NN}(n,k-1)\vert}{\vert {\cal TC}({n,k-1})\vert} \cdot \left[
1+\frac{\sqrt{\pi}}{2\sqrt n}
+
o\!\left(\frac{1}{\sqrt{n}}\right)
\right]
\\ 
&& -\frac{3(2n+k-3)(n-k)}{2n(n-k+1)}\cdot  \left[
1+\frac{\sqrt{\pi}}{2\sqrt n}
+
o\!\left(\frac{1}{\sqrt{n}}\right)
\right]+3\\
&\geq & \frac{\vert {\cal NN}(n,k-1)\vert}{\vert {\cal TC}({n,k-1})\vert} \left(1-\frac{k+2}{n-k+1} \right)\left[
1+\frac{\sqrt{\pi}}{2\sqrt n}
+
o\!\left(\frac{1}{\sqrt{n}}\right)
\right] + 3\\
&& -3\left(1+\frac{k-3}{2n} \right)\cdot \left[
1+\frac{\sqrt{\pi}}{2\sqrt n}
+o\!\left(\frac{1}{\sqrt{n}}\right)
\right]\\
&=&\frac{\vert {\cal NN}(n,k-1)\vert}{\vert {\cal TC}({n,k-1})\vert} \left[
1+\frac{\sqrt{\pi}}{2\sqrt n}
+
o\!\left(\frac{1}{\sqrt{n}}\right)
\right] -\frac{3\sqrt{\pi}}{2\sqrt{n}} +o\left(\frac{1}{\sqrt{n}}\right).
\end{eqnarray*}
This implies:
\[ \frac{\vert {\cal NN}({n,k})\vert}{\vert {\cal TC}({n,k})\vert}
\geq
1-\frac{\sqrt{\pi}k}{\sqrt n}
+
o\!\left(\frac{k}{\sqrt n}\right),
\]
shown in Proposition B2 in the appendix. 
Therefore, we have the theorem.
\begin{theorem}
\label{thm4}
Let ${\cal NN}(n,k)$ denote the set of normal networks with $k$ reticulations on a set of $n$ taxa. 
If $1\leq k=o(\sqrt{n})$, then
\[
\vert {\cal NN}(n,k)\vert
\sim
\binom{n}{k}\frac{2^{n+k-1/2}\,n^{n-1+k}}{e^{n}}.
\]
\end{theorem}

\section{Conclusions}

It remains unclear how the techniques developed for the asymptotic counting of BPNs in 
\cite{chang2024,fuchs2024asymptotic,HaoYu2026_JCB,mansouri2022counting} can be applied when $k$ is non-constant but satisfies $k=o(\sqrt{n})$. 
In this paper, we studied the generation of binary phylogenetic networks by edge insertion.
By bounding the contribution of networks with exceptional local structures and combining these bounds with known asymptotic formulas for tree-child networks, we derived the same asymptotic counting result for BPNs, tree-child networks, and normal networks with $k$ reticulations on $n$ taxa when $k=o(\sqrt{n})$.
This range is essentially tight for our asymptotic formula, as the result does not extend in general to the regime $k\geq \sqrt{n}$.

This study raises the following research problem:
\begin{quote}
 Does the same asymptotic counting result hold for  the class of galled networks when $1\leq k=o(\sqrt{n})$?
\end{quote}
Here, a \emph{galled network} is a BPN in which the two parents of each reticulation node lie in  the same tree component. We note that the same asymptotic counting result  holds for galled networks 
in the more restricted range 
$1\leq k=o\left(n^{1/3}\right)$
(see the appendix, where its proof is obtained by modifying the argument in \cite{chang2024}).

\section*{Acknowledgments}
This work was financially supported by the Singapore MOE Academic Research
Fund Tier 1 [A-8001951-00-00].
The authors thank Michael Fuchs for his comments on Proposition~4.1 and Mike Steel for helpful communication regarding their work on normal networks.



\clearpage
\appendix

\markboth{Appendix}{Appendix} 

\begin{center} {\large  {\bf APPENDIX} }
\end{center}

\section*{A. Proof of Proposition 4.1}

 We  have
\begin{equation}\nonumber
    \vert{\cal TC}(n, k)\vert 
=
(n-k+1)\vert{\cal TC}(n, k-1)\vert 
+
\frac{n(2n+k-3)}{n-k}\vert{\cal TC}(n-1, k)\vert ,
\end{equation}
for $1\le k\le n-1$, and 
\[
\vert{\cal TC}(n, 0)\vert =(2n-3)!!=\frac{(2n-2)!}{2^{n-1}(n-1)!},
\qquad
\vert{\cal TC}(k, k)\vert =0.
\]

We extend the asymptotic count of ${\cal TC}(n, k)$ for constant $k$ in Pons and Batle's paper to the range $k=o(\sqrt{n})$. By convention, we set $0!!=(-1)!!=1$.
\\

\noindent {\bf Proposition.}
 For \(1\leq k=o(\sqrt n)\), 

\begin{eqnarray*}
\vert{\cal TC}(n, k)\vert 
&=& 
\binom{n}{k}\sqrt2 e^{-n}(2n)^{n+k-1} \\
&& \times \Bigg[
1
-\frac{\sqrt{\pi}k}{2\sqrt n}
+
\frac{14k^2-26k+11}{24n} 
-\frac{\sqrt{\pi}k(31k^2-93k+70)}{192n^{3/2}} \\
& & 
+
\frac1{n^2}
\left(
\frac{725}{6048}k^4
-\frac{599}{1008}k^3
+\frac{1579}{1512}k^2
-\frac{179}{224}k
+\frac{265}{1152}
\right) 
+
O\left(\frac{k^5}{n^{5/2}}\right)
\Bigg].
\end{eqnarray*}
\begin{proof} 
Define
\[
A_{n,k}
=
\frac{ \vert{\cal TC}(n, k)\vert }
{\binom{n}{k}(2(n+k)-3)!!},
\]
for any $1\leq k<n$.

We first derive the recurrence for \(A_{n,k}\).  By definition, 
\[
\vert{\cal TC}(n, k-1)\vert 
=
\binom{n}{k-1}(2(n+k)-5)!!A_{n,k-1},
\]
and we have
\[
\begin{aligned}
\frac{(n-k+1)\vert{\cal TC}(n, k-1)\vert }
{\binom{n}{k}(2(n+k)-3)!!}
&=
(n-k+1)
\frac{\binom{n}{k-1}}{\binom{n}{k}}
\frac{(2(n+k)-5)!!}{(2(n+k)-3)!!}
A_{n,k-1}  
&=
\frac{k}{2(n+k)-3}A_{n,k-1}.
\end{aligned}
\]
Similarly, by definition, 
\[
\vert{\cal TC}(n-1, k)\vert 
=
\binom{n-1}{k}(2(n+k)-5)!!A_{n-1,k},
\]
and therefore 
\[
\begin{aligned}
\frac{
n(2n+k-3)\vert{\cal TC}(n-1, k)\vert 
}
{(n-k) \binom{n}{k}(2(n+k)-3)!!}
&=
\frac{2n+k-3}{2(n+k)-3}A_{n-1,k}.
\end{aligned}
\]
Hence, we obtain:
\[
A_{n,k}
=
\frac{k}{2(n+k)-3}A_{n,k-1}
+
\frac{2n+k-3}{2(n+k)-3}A_{n-1,k}.
\tag{E1}
\]
The two coefficients in Equation~(E1) are non-negative and satisfy
\[
\frac{k}{2(n+k)-3}
+
\frac{2n+k-3}{2(n+k)-3}
=
1.
\]
Thus Equation~(E1) is a weighted-average recurrence.

For \(i\ge0\), define
\[
H_i(n,k)
=
\binom{k}{i}
\frac{(2(n+k)-i-3)!!}{(2(n+k)-3)!!},
\]
where \(\binom{k}{i}=0\) if \(i>k\). We claim that each \(H_i\) satisfies the same recurrence as \(A_{n,k}\), namely
\[
H_i(n,k)
=
\frac{k}{2(n+k)-3}H_i(n,k-1)
+
\frac{2n+k-3}{2(n+k)-3}H_i(n-1,k).
\tag{E2}
\]
Indeed,
\[
H_i(n,k-1)
=
\binom{k-1}{i}
\frac{(2(n+k)-i-5)!!}{(2(n+k)-5)!!},
\]
and
\[
H_i(n-1,k)
=
\binom{k}{i}
\frac{(2(n+k)-i-5)!!}{(2(n+k)-5)!!}.
\]
Hence the right-hand side of (E2) is
\[
\left[
\frac{k}{2(n+k)-3}\binom{k-1}{i}
+
\frac{2n+k-3}{2(n+k)-3}\binom{k}{i}
\right]
\frac{(2(n+k)-i-5)!!}{(2(n+k)-5)!!}.
\]
Since
\[
k\binom{k-1}{i}
=
(k-i)\binom{k}{i},
\]
the sum in the square bracket becomes
\[
\frac{(k-i)+(2n+k-3)}{2(n+k)-3}\binom{k}{i}
=
\frac{2(n+k)-i-3}{2(n+k)-3}\binom{k}{i}.
\]
Therefore the right-hand side of (E2) equals
\[
\binom{k}{i}
\frac{2(n+k)-i-3}{2(n+k)-3}
\frac{(2(n+k)-i-5)!!}{(2(n+k)-5)!!}
=
\binom{k}{i}
\frac{(2(n+k)-i-3)!!}{(2(n+k)-3)!!}
=
H_i(n,k).
\]
This proves (E2).

Now define
\[
S_4(n,k)
=
H_0(n,k)
-H_1(n,k)
+\frac13H_2(n,k)
+\frac{17}{8}H_3(n,k)
-\frac{283}{63}H_4(n,k).
\tag{E3}
\]
Because every \(H_i\) satisfies (E1), the linear combination \(S_4\) also satisfies (E1).

We now show that \(S_4(n,k)\) approximates \(A_{n,k}\) with a rigorous error bound. Let
\[
E_{n,k} = A_{n,k}-S_4(n,k).
\]
Since both \(A_{n,k}\) and \(S_4(n,k)\) satisfy (E1), the error satisfies
\[
E_{n,k} = \frac{k}{2(n+k)-3}E_{n,k-1} + \frac{2n+k-3}{2(n+k)-3}E_{n-1,k}.
\tag{E4}
\]

We control \(E_{n,k}\) by
\[
H_5(n,k)
=
\binom{k}{5}
\frac{(2(n+k)-8)!!}{(2(n+k)-3)!!}.
\]
By (E2), \(H_5\) satisfies the same recurrence:
\[
H_5(n,k)
=
\frac{k}{2(n+k)-3}H_5(n,k-1)
+
\frac{2n+k-3}{2(n+k)-3}H_5(n-1,k).
\tag{E5}
\]

We first check the boundary. We can verify that
\[
E_{n,k}=0, 
\qquad
0\le k\le4.
\tag{E6}
\]

On the diagonal \(n=k\), we have \(A_{k,k}=0\), and therefore
\[
E_{k,k}=-S_4(k,k).
\]
For each fixed \(i\le4\), the standard double-factorial ratio estimate gives
\[
H_i(k,k)
=
\binom{k}{i}
\frac{(4k-i-3)!!}{(4k-3)!!}
=
O(k^{i/2}).
\]
Thus
\[
S_4(k,k)=O(k^2).
\tag{E7}
\]
On the other hand,
\[
H_5(k,k)
=
\binom{k}{5}
\frac{(4k-8)!!}{(4k-3)!!}
=O\left(k^5k^{-5/2}\right)
= O\left(k^{5/2}\right).
\tag{E8}
\]
It follows that there exists a constant \(C>0\) such that
\[
|E_{k,k}|
\le
C H_5(k,k)
\qquad(k\ge5).
\tag{E9}
\]
Together with (E6), this gives the necessary boundary control.

We prove by induction on \(n+k\) that
\[
|E_{n,k}|
\le
C H_5(n,k), \; 0\le k\le n.
\tag{E10}
\]
The cases \(0\le k\le4\) follow from (E6), as \(H_5(n,k)=0\). The diagonal cases \(n=k\) follow from (E9). Now suppose \(5\le k\le n-1\), and assume that (E10) has been proved for \((n,k-1)\) and \((n-1,k)\). By (E4), the non-negativity of the coefficients, and the induction hypothesis,
\[
\begin{aligned}
|E_{n,k}|
&\le
\frac{k}{2(n+k)-3}|E_{n,k-1}|
+
\frac{2n+k-3}{2(n+k)-3}|E_{n-1,k}| \\
&\le
C\left[
\frac{k}{2(n+k)-3}H_5(n,k-1)
+
\frac{2n+k-3}{2(n+k)-3}H_5(n-1,k)
\right].
\end{aligned}
\]
By (E5), the bracket is exactly \(H_5(n,k)\). Hence
\[
|E_{n,k}|\le C H_5(n,k).
\]
This proves (E10).

Since
\[
H_5(n,k)
=
\binom{k}{5}
\frac{(2(n+k)-8)!!}{(2(n+k)-3)!!}
=
O\left(\frac{(k+1)^5}{(n+k)^{5/2}}\right),
\]
we have
\[
A_{n,k}
=
S_4(n,k)
+
O\left(\frac{(k+1)^5}{(n+k)^{5/2}}\right).
\tag{E11}
\]

Since \(k=o(\sqrt n)\),  using Stirling's approximation, we have uniformly in this range,
\[
H_0(n,k)=1,
\]
\[
H_1(n,k)
=
\frac{\sqrt{2\pi} \;k}{2}\frac{1}{\sqrt{2n}}
+
\frac{\sqrt{2\pi}\; k(5-4k)}{8}\left(\frac{1}{\sqrt{2n}}\right)^3
+
{O\left(\frac{k^3}{n^{5/2}} \right)},
\]
\[
H_2(n,k)
=
\binom{k}{2}\left(\frac{1}{\sqrt{2n}}\right)^2
-
(2k-3)\binom{k}{2}\left(\frac{1}{\sqrt{2n}}\right)^4
+
{O\left(\frac{k^4}{n^3}\right)},
\]
\[
H_3(n,k)
=
\frac{\sqrt{2\pi}}{2}\binom{k}{3}\left(\frac{1}{\sqrt{2n}}\right)^3
+
{O\left(\frac{k^4}{n^{5/2}}\right)},
\]
and
\[
H_4(n,k)
=
\binom{k}{4}\left(\frac{1}{\sqrt{2n}}\right)^4
+
{O\left(\frac{k^5}{n^3}\right)}.
\]
Substituting these expansions into (E3) and using (E11), we get
\[
\begin{aligned}
A_{n,k}
={}&
1
-\frac{\sqrt{\pi}}{2}\left(\frac{k}{\sqrt{n}}\right)
+\frac13\binom{k}{2}\left(\frac{1}{\sqrt{2n}}\right)^2 \\
&+
\frac{\sqrt{\pi}}{4}
\left[
k^2-\frac{5k}{4}+\frac{17}{8}\binom{k}{3}
\right]\left(\frac{1}{\sqrt{n}}\right)^3 \\
&-
\left[
\frac{2k-3}{12}\binom{k}{2}
+\frac{283}{252}\binom{k}{4}
\right]\left(\frac{1}{\sqrt{n}}\right)^4 \\
&+
{O\left(\frac{k^5}{n^{5/2}}\right)}.
\end{aligned}
\tag{E12}
\]

Next we expand the double-factorial factor. Stirling's formula gives
\[
(2(n+k)-3)!!
=
\frac{\sqrt{2} (2(n+k))^{n+k-1}}{e^{n+k}}
\left[
1+\frac{11}{24(n+k)}
+\frac{265}{1152(n+k)^2}
+
O\left(\frac{1}{(n+k)^{3}}\right)
\right].
\tag{E13}
\]
Note that
\[
(2(n+k)-3)!!
=
\sqrt2 e^{-n}(2n)^{(n+k)-1}D_{n,k},
\]
where
\[
D_{n,k}
=
e^{-k}
\left(1+\frac{k}{n}\right)^{(n+k)-1}
\left[
1+\frac{11}{24(n+k)}
+\frac{265}{1152(n+k)^2}
+
O((n+k)^{-3})
\right].
\]
Expanding in powers of \((2n)^{-1/2}\), uniformly for \(k=o(\sqrt n)\), gives
\[
\begin{aligned}
D_{n,k}
={}&
1
+
\left(k^2-2k+\frac{11}{12}\right)\left(\frac{1}{\sqrt{2n}}\right)^2 \\
&+
\left(
\frac12k^4
-\frac83k^3
+\frac{59}{12}k^2
-\frac{11}{3}k
+\frac{265}{288}
\right)\left(\frac{1}{\sqrt{2n}}\right)^4 \\
&+
{O\left(\frac{k^5}{n^{5/2}}\right)}.
\end{aligned}
\tag{E14}
\]

Combining
\[
\vert{\cal TC}(n, k)\vert 
=
\binom{n}{k}(2(n+k)-3)!!A_{n,k}
\]
with (E12) and (E14), we obtain the formula in the proposition.
\end{proof}

\section*{B: Two Facts Used in Analysis of Normal Networks}

\noindent {\bf Proposition B1}. For each $k$ such that 
$1\leq k=o(\sqrt{n})$, 
\[
\frac{|{\cal TC}(n,k-1)|}{|{\cal TC}(n,k)|}
=
\frac{k}{2n(n-k+1)}
\left[
1+\frac{\sqrt{\pi}}{2\sqrt n}
+o\left(\frac1{\sqrt n}\right)
\right].
\]
\begin{proof}
Recall that in Appendix A we defined 
\[
A_{n,k}
=
\frac{\vert {\cal TC}(n, k)\vert}
{\binom{n}{k}(2n+2k-3)!!},
\]
and
\[
H_i(n,k)
=
\binom{k}{i}
\frac{(2n+2k-i-3)!!}{(2n+2k-3)!!}.
\]
Now we define 
\[F_{n,k}=A_{n,k}-\Big(H_0(n,k)-H_1(n,k)\Big).\]

By an argument similar to the one used in Appendix A, we obtain:

\[
|F_{n,k}|
\leq C H_2(n,k),
\tag{E15}
\]
for some constant $C>0$, and  that \(F_{n,k}\) satisfies the same recurrence
\[
F_{n,k}
=
\frac{k}{2n+2k-3}F_{n,k-1}
+
\frac{2n+k-3}{2n+2k-3}F_{n-1,k}.
\tag{E16}
\]
We first prove a difference estimate for the error term. We shall prove that, uniformly for \(k=o(\sqrt n)\),
\[
F_{n,k}-F_{n,k-1}
=
o\left(\frac1{\sqrt n}\right).
\tag{E17}
\]

We first prove a stronger estimate. We claim that there exists a constant
\(D>0\) such that, for \(1\leq k\leq n-1\),
\[
\begin{aligned}
\left|F_{n-1,k}-F_{n,k-1}\right|
\leq D\Big\{
& H_2(n-1,k)-H_2(n,k-1) \\
&+H_4(n-1,k)-H_4(n,k-1)
\Big\}.
\end{aligned}
\tag{E18}
\]

Applying (E16) to \(F_{n-1,k}\) and \(F_{n,k-1}\),
we obtain:
\[
\begin{aligned}
F_{n-1,k}-F_{n,k-1}
={}&
\frac{k-1}{2n+2k-5}
\left(F_{n-1,k-1}-F_{n,k-2}\right) \\
&+
\frac{2n+k-5}{2n+2k-5}
\left(F_{n-2,k}-F_{n-1,k-1}\right).
\end{aligned}
\tag{E19}
\]
Note that the two coefficients are nonnegative and their sum is less than \(1\) in (E19).

Similarly, since each \(H_i(n,k)\) satisfies the same recurrence, we have
\[
\begin{aligned}
H_i(n-1,k)-H_i(n,k-1)
={}&
\frac{k-1}{2n+2k-5}
\left(H_i(n-1,k-1)-H_i(n,k-2)\right) \\
&+
\frac{2n+k-5}{2n+2k-5}
\left(H_i(n-2,k)-H_i(n-1,k-1)\right).
\end{aligned}
\tag{E20}
\]
Moreover, by the definition of \(H_i\),
\[
H_i(n-1,k)-H_i(n,k-1)
=
\binom{k-1}{i-1}
\frac{(2n+2k-i-5)!!}{(2n+2k-5)!!}.
\tag{E21}
\]
In particular, these quantities are nonnegative for \(i=2\) and \(i=4\).

We now prove (E18) by induction on \(n+k\). If \(k=0\) or \(k=1\), then
\(F_{n,k}=0\), and hence (E18) holds.

Next consider the boundary case \(n=k+1\). By (E15),
\[
\begin{aligned}
\left|F_{k,k}-F_{k+1,k-1}\right|
&\leq
\vert F_{k,k}|+|F_{k+1,k-1}\vert  \\
&\leq
C H_2(k,k)+C H_2(k+1,k-1).
\end{aligned}
\]
Using the standard double-factorial ratio estimate, both terms on the
right-hand side are \(O(k)\). On the other hand, by (E21),
\[
H_4(k,k)-H_4(k+1,k-1)
=
\binom{k-1}{3}
\frac{(4k-7)!!}{(4k-3)!!}=O( k).
\]
Thus, after increasing \(C\) if necessary, (E18) holds on the boundary
\(n=k+1\). The finitely many small values of \(k\) are absorbed into the
same constant.

Now suppose \(2\leq k\leq n-2\), and assume that (E18) has already been
proved for the two pairs appearing on the right-hand side of (E19). By
(E19), the nonnegativity of the coefficients, and the induction hypothesis,
we obtain
\[
\begin{aligned}
\left|F_{n-1,k}-F_{n,k-1}\right|
\leq{}&
\frac{k-1}{2n+2k-5}
\left|F_{n-1,k-1}-F_{n,k-2}\right| \\
&+
\frac{2n+k-5}{2n+2k-5}
\left|F_{n-2,k}-F_{n-1,k-1}\right| \\
\leq{}&
D\frac{k-1}{2n+2k-5}
\Big\{
H_2(n-1,k-1)-H_2(n,k-2) \\
&\hspace{3.5cm}
+H_4(n-1,k-1)-H_4(n,k-2)
\Big\} \\
&+
D\frac{2n+k-5}{2n+2k-5}
\Big\{
H_2(n-2,k)-H_2(n-1,k-1) \\
&\hspace{3.5cm}
+H_4(n-2,k)-H_4(n-1,k-1)
\Big\}.
\end{aligned}
\]
By (E20), the last expression is at most
\[
D\Big\{
H_2(n-1,k)-H_2(n,k-1)
+
H_4(n-1,k)-H_4(n,k-1)
\Big\}.
\]
This finishes the proof of (E18).

Now return to \(F_{n,k}-F_{n,k-1}\). From (E16), we have
\[
\begin{aligned}
F_{n,k}-F_{n,k-1}
&=
\frac{k}{2n+2k-3}F_{n,k-1}
+
\frac{2n+k-3}{2n+2k-3}F_{n-1,k}
-
F_{n,k-1}  \\
&=
\frac{2n+k-3}{2n+2k-3}
\left(F_{n-1,k}-F_{n,k-1}\right).
\end{aligned}
\tag{E22}
\]
Combining (E18) and (E22), we get
\[
\begin{aligned}
\left|F_{n,k}-F_{n,k-1}\right|
\leq D\Big\{
& H_2(n-1,k)-H_2(n,k-1) \\
&+H_4(n-1,k)-H_4(n,k-1)
\Big\}.
\end{aligned}
\tag{E23}
\]
Using (E21), we have
\[
H_2(n-1,k)-H_2(n,k-1)
=
(k-1)\frac1{2n+2k-5}
=
O\left(\frac{k}{n+k}\right),
\tag{E24}
\]
and
\[
H_4(n-1,k)-H_4(n,k-1)
=
\binom{k-1}{3}
\frac{(2n+2k-9)!!}{(2n+2k-5)!!}
=
O\left(\frac{k^3}{(n+k)^2}\right).
\tag{E25}
\]
Since \(k=o(\sqrt n)\), we have
\[
\frac{k+1}{n+k}
=
o\left(\frac1{\sqrt n}\right),
\]
and
\[
\frac{(k+1)^3}{(n+k)^2}
=
\frac1{\sqrt n}
\left(\frac{k+1}{\sqrt n}\right)^3
O(1)
=
o\left(\frac1{\sqrt n}\right).
\]
Therefore
\[
F_{n,k}-F_{n,k-1}
=
o\left(\frac1{\sqrt n}\right)
\]
uniformly for \(k=o(\sqrt n)\).

Now, we are ready to prove Proposition B1. By definition,
\[
|{\cal TC}(n,k)|
=
\binom{n}{k}(2n+2k-3)!!A_{n,k}.
\]
Therefore, 
\[
\begin{aligned}
\frac{|{\cal TC}(n,k-1)|}{|{\cal TC}(n,k)|}
&=
\frac{\binom{n}{k-1}}{\binom{n}{k}}
\frac{(2n+2k-5)!!}{(2n+2k-3)!!}
\frac{A_{n,k-1}}{A_{n,k}}  \\
&=
\frac{k}{(n-k+1)(2n+2k-3)}
\frac{A_{n,k-1}}{A_{n,k}}.
\end{aligned}
\tag{E26}
\]

We now estimate \(A_{n,k-1}/A_{n,k}\). From the simplified approximation,
\[
A_{n,k}
=
H_0(n,k)-H_1(n,k)+F_{n,k}
=
1-H_1(n,k)+F_{n,k}.
\]
By the difference estimate proved above,
\[
F_{n,k}-F_{n,k-1}
=
o\left(\frac1{\sqrt n}\right)
\tag{E27}
\]
uniformly for \(k=o(\sqrt n)\).

Next,
\[
H_1(n,k)
=
k
\frac{(2n+2k-4)!!}{(2n+2k-3)!!}.
\]
Using the standard double-factorial estimate,
\[
\frac{(2n+2k-4)!!}{(2n+2k-3)!!}
=
\frac{\sqrt{\pi}}{2\sqrt{n+k}}
+
O((n+k)^{-3/2}),
\]
we obtain
\[
H_1(n,k)-H_1(n,k-1)
=
\frac{\sqrt{\pi}}{2\sqrt n}
+
o\left(\frac1{\sqrt n}\right).
\tag{E28}
\]
Hence
\[
\begin{aligned}
A_{n,k-1}-A_{n,k}
&=
H_1(n,k)-H_1(n,k-1)
+
F_{n,k-1}-F_{n,k} \\
&=
\frac{\sqrt{\pi}}{2\sqrt n}
+
o\left(\frac1{\sqrt n}\right).
\end{aligned}
\tag{E29}
\]

Moreover, from
\[
A_{n,k}
=
1-H_1(n,k)+F_{n,k}
\]
and
\[
H_1(n,k)=O\left(\frac{k}{\sqrt n}\right)=o(1),
\qquad
F_{n,k}=O\left(\frac{(k+1)^2}{n}\right)=o(1),
\]
we have
\[
A_{n,k}=1+o(1).
\tag{E30}
\]
Therefore
\[
\frac1{A_{n,k}}=1+o(1).
\tag{E31}
\]
Combining (E29) and (E31), we get
\[
\frac{A_{n,k-1}}{A_{n,k}}
=
1+
\frac{A_{n,k-1}-A_{n,k}}{A_{n,k}}
=
1+
\frac{\sqrt{\pi}}{2\sqrt n}
+
o\left(\frac1{\sqrt n}\right).
\tag{E32}
\]

It remains to expand the first fraction on the right-hand side of (E26).  Since \(k=o(\sqrt n)\),
\[
\frac1{2n+2k-3}
=
\frac1{2n}
\left[
1+o\left(\frac1{\sqrt n}\right)
\right].
\tag{E33}
\]
Substituting (E32) and (E33) into (E26), we obtain
\[
\frac{|{\cal TC}(n,k-1)|}{|{\cal TC}(n,k)|}
=
\frac{k}{2n(n-k+1)}
\left[
1+\frac{\sqrt{\pi}}{2\sqrt n}
+o\left(\frac1{\sqrt n}\right)
\right].
\]
This completes the proof.
\end{proof}

\noindent {\bf Proposition B2.}
Let $a(k)$ be a non-negative function of the integer $k$ such that $a(0)=1$ and
\[
a(k)\geq a(k-1)\left(1+\frac{\sqrt{\pi}}{2\sqrt{n}}\right)
 - \frac{3\sqrt{\pi}}{2\sqrt{n}} + o\left(\frac{1}{\sqrt{n}}\right),
\]
where the $o(1/\sqrt n)$ term is uniform for $1\leq k=o(\sqrt n)$. Then
\[
a(k)\geq 1-\frac{\sqrt{\pi}k}{\sqrt n}
+
o\!\left(\frac{k}{\sqrt n}\right)
\]
for each $k=o(\sqrt{n})$.

\begin{proof}
Set
\[
s_n=\frac{\sqrt{\pi}}{2\sqrt n}.
\]
Then the recurrence can be written as
\[
a(k)\geq (1+s_n)a(k-1)-3s_n+r_n,
\]
where
\[
r_n=o\left(\frac{1}{\sqrt n}\right)
\]
uniformly for $1\leq k=o(\sqrt n)$.

Subtracting \(3\) from both sides gives
\[
a(k)-3
\geq
(1+s_n)\bigl(a(k-1)-3\bigr)+r_n.
\]
Iterating this inequality from \(1\) to \(k\), and using \(a(0)=1\), we obtain
\[
a(k)-3
\geq
(1+s_n)^{k-1}(a(1)-3)
+
r_n\sum_{j=0}^{k-2}(1+s_n)^j.
\]
Since \(a(0)-3=-2\), this becomes
\[
a(k)
\geq
3-2(1+s_n)^{k-1}
+
r_n\sum_{j=0}^{k-2}(1+s_n)^j.
\]

Because \(k=o(\sqrt n)\), we have \(ks_n=o(1)\). Hence
\[
(1+s_n)^{k-1}
=
1+(k-1)s_n+o\left(\frac{k}{\sqrt n}\right).
\]
Also,
\[
\sum_{j=0}^{k-2}(1+s_n)^j
=
O(k),
\]
and therefore
\[
r_n\sum_{j=0}^{k-2}(1+s_n)^j
=
o\left(\frac{k}{\sqrt n}\right).
\]
It follows that
\[
a(k)
\geq
3-2\left[
1+(k-1)s_n+o\left(\frac{k}{\sqrt n}\right)
\right]
+
o\left(\frac{k}{\sqrt n}\right).
\]
Thus
\[
a(k)
\geq
1-2(k-1)s_n
+
o\left(\frac{k}{\sqrt n}\right).
\]
Since
\[
2s_n=\frac{\sqrt{\pi}}{\sqrt n},
\]
replacing $(k-1)$ with $k$, we obtain
\[
a(k)
\geq
1-\frac{\sqrt{\pi}(k-1)}{\sqrt n}
+
o\left(\frac{k}{\sqrt n}\right) 
\geq 1-\frac{\sqrt{\pi}k}{\sqrt n}
+
o\left(\frac{k}{\sqrt n}\right).
\]
This proves the proposition.
\end{proof}

\section*{C. Asymptotic Analysis of Galled Networks}

\subsection*{C1: Three Lemmas} 

The following result is known in the literature.\\

\noindent {\bf Lemma C1.}
{\it 
For $k>0$,
\[
\sum_{i=0}^{\infty} \binom{2i+k}{i}\frac{1}{(2i+k)2^{2i}} = \frac{2^k}{k}.
\]
}

\begin{proof}
Use the identity
\[
\frac{1}{2i+k} = \int_0^1 t^{2i+k-1}\,dt.
\]
Then
\[
\begin{aligned}
S 
&= \sum_{i=0}^{\infty} \binom{2i+k}{i}\frac{1}{(2i+k)4^i} 
= \int_0^1 t^{k-1} \sum_{i=0}^{\infty} \binom{2i+k}{i} \left(\frac{t^2}{4}\right)^i dt.
\end{aligned}
\]

Using the generating function
\[
\sum_{i=0}^{\infty} \binom{2i+k}{i} x^i
=
\frac{1}{\sqrt{1-4x}}
\left(\frac{1-\sqrt{1-4x}}{2x}\right)^k,
\]
with $x = \frac{t^2}{4}$, we obtain
\[
S = \int_0^1
\frac{t^{k-1}}{\sqrt{1-t^2}}
\left(\frac{2}{1+\sqrt{1-t^2}}\right)^k dt.
\]

Let $t=\sin\theta$. Then
\[
S = \int_0^{\pi/2}
\sin^{k-1}\theta
\left(\frac{2}{1+\cos\theta}\right)^k d\theta.
\]

Now set $\theta = 2u$. Since $1+\cos\theta = 2\cos^2 u$ and $\sin\theta = 2\sin u \cos u$, we get
\[
\begin{aligned}
S
&= 2^k \int_0^{\pi/4}
\frac{\sin^{k-1}u}{\cos^{k+1}u}\,du 
= 2^k \int_0^{\pi/4} \tan^{k-1}u\,\sec^2 u\,du \\
&= 2^k \left[\frac{\tan^k u}{k}\right]_0^{\pi/4}
= \frac{2^k}{k}.
\end{aligned}
\]
\end{proof}


\noindent {\bf Lemma C2.}
{\it 
For any $t\geq 2k$, 
\[
\sum_{i=t}^{\infty}\binom{2i+k}{i}\frac{1}{(2i+k)2^{2i}}
\leq \frac{2^k}{\sqrt{\pi t}}.
\]
}

\begin{proof} 
Since \[ \sqrt{2\pi m} \left(\frac{m}{e}\right)^m 
\leq m! \leq \sqrt{2\pi m} \left(\frac{m}{e}\right)^m e^{\frac{1}{12m},}\]
for any integer $m>0$, 
\begin{eqnarray*}
\frac{(2i+k-1)!}{i!(i+k)!} &\leq & \frac{
\sqrt{2\pi(2i+k)}\left(\frac{2i+k}{e}\right)^{2i+k} e^{\frac{1}{12(2i+k)}}
}{
\sqrt{2\pi i}(2i+k)\left(\frac{i}{e}\right)^i
\cdot
\sqrt{2\pi(i+k)}\left(\frac{i+k}{e}\right)^{i+k}
}\\
&=& \frac{
(2i+k)^{2i+k-1/2} e^{\frac{1}{12(2i+k)}}
}{
\sqrt{2\pi}\cdot i^{i+1/2}
\cdot (i+k)^{i+k+1/2}
}\\
&=& \frac{2^{2i+k-1}
(i+k/2)^{2i+k-1/2} }
{
\sqrt{\pi}\cdot i^{i+1/2}
\cdot (i+k)^{i+k+1/2}
}e^{\frac{1}{12(2i+k)}} \\
&=& \frac{2^{2i+k-1}\left(1+\frac{k}{2i}\right)^{i-1}
\left(1-\frac{k}{(2i+2k)}\right)^{i+k+1/2}}
{
\sqrt{\pi}\cdot i^{3/2}
}e^{\frac{1}{12(2i+k)}}\\
&\leq&  \frac{2^{2i+k-1}e^{\frac{k}{2}-\frac{k}{2i}}
e^{-\frac{k}{2}-\frac{k}{4(i+k)}}
}{
\sqrt{\pi}\cdot i^{3/2}
} e^{\frac{1}{12(2i+k)}} \\
&=&  \frac{2^{2i+k-1}
}{
\sqrt{\pi}\cdot i^{3/2}
} e^{\frac{1}{12(2i+k)}-\frac{k}{2i}-\frac{k}{4(i+k)}} \\
&\leq & \frac{2^{2i+k-1}
}{\sqrt{\pi}\cdot i^{3/2}
} \left( 1+\frac{1}{12(2i+k)}-\frac{k}{2i}-\frac{k}{4(i+k)} +
\left(\frac{1}{12(2i+k)}-\frac{k}{2i}-\frac{k}{4(i+k)}\right)^2\right)\\
&\leq & \frac{2^{2i+k-1}
}{\sqrt{\pi}\cdot i^{3/2}
} \left( 1- \frac{3k}{4(i+k)}
+ \left(\frac{k}{i}\right)^2\right)\\
&\leq & \frac{2^{2i+k-1}
}{\sqrt{\pi}\cdot i^{3/2}}, 
\end{eqnarray*}
where the second and third inequalities follow from the fact
$1+m\leq e^m \leq 1+m+m^2$ for $1\geq m\geq -1$;
and the last inequality is derived from the fact that
$$ \frac{3k}{4(i+k)} \geq \left(\frac{k}{i}\right)^2$$
for any $i \geq 2k$.


Therefore, for $t\geq 2k$,
\[
\sum_{i=t}^{\infty}\frac{(2i+k-1)!}{i!(i+k)!2^{2i}}
\leq
\frac{2^{k-1}}{\sqrt{\pi}} 
\sum_{i=t}^{\infty} i^{-3/2}
=
\frac{2^k}{\sqrt{\pi t}}.
\]
\end{proof}

\noindent {\bf Lemma C3.}
{\it 
For any $i>0$ such that $i+k< \sqrt{n/2}$, 
\begin{eqnarray*}
\frac{(2n-2i-2)!(n-k)!}{(n-k-i-1)!(n-k-i)!} 
\geq \frac{2^{2n-3/2}n^{n+k-1}}{2^{2i}e^n} 
\left(1 -
\frac{i(4k-3)+(2k-1)^2}{2(n-i-1)}
 - \frac{7k+i+2}{6n}\right)
\end{eqnarray*}
}
\begin{proof} 
Let the left-hand side of the Inequality in the lemma be A.
By Inequality~(\ref{factorial_to_power}), 
  \begin{eqnarray*}
A&= &  \frac{(n-k-i)(2n-2i-2)!(n-k)!}{(n-k-i)!(n-k-i)!}\\
&\geq & \frac{(n-k)^{n-k+1/2}2^{2n-2i-3/2}(n-i-1)^{2n-2i-3/2}}{e^{n+k-2} (n-k-i)^{2n-2k-2i} }e^{-\frac{1}{6(n-k-i)}}\\
&=& \frac{(n-k)^{n-k+1/2}2^{2n-2i-3/2}(n-i-1)^{2n-2k-2i}(n-i-1)^{2k-3/2}}{e^{n+k-2} (n-k-i)^{2n-2k-2i} }e^{-\frac{1}{6(n-k-i)}}\\
&=& \frac{2^{2n-2i-2}n^{n+k-1}\left(1-\frac{k}{n}\right)^{n-k+1/2}\left(1-\frac{i+1}{n}\right)^{2k-3/2}}{e^{n+k-2}\left(1-\frac{k-1}{n-i-1}\right)^{2n-2k-2i}}e^{-\frac{1}{6(n-k-i)}}\\
&\geq & \frac{2^{2n-2i-2}n^{n+k-1}\left(1-\frac{k}{n}\right)^{n-k+1/2}\left(1-\frac{i+1}{n}\right)^{2k-3/2}}{e^{n+k-2}e^{-2(k-1)+\frac{2(k-1)^2}{n-i-1}}}e^{-\frac{1}{6(n-k-i)}}\\
&=& \frac{2^{2n-2i-2}n^{n+k-1}\left(1-\frac{k}{n}\right)^{n-k+1/2}\left(1-\frac{i+1}{n}\right)^{2k-3/2}}{e^{n-k}}e^{-\frac{2(k-1)^2}{n-i-1}-\frac{1}{6(n-k-i)}}.\\
\end{eqnarray*}
Note that
\[
   e^{-x/(1-x)} \leq 1-x \leq e^{-x}
   \tag{E34}
\]
for any  $0\leq x\leq 1$.

Applying the first inequality in (E34) to 
$1-k/n$ and $1-(i+1)/n$, we obtain: 
\begin{eqnarray*}
A &\geq & \frac{2^{2n-2i-2}n^{n+k-1}
e^{-k-\frac{k}{2(n-k)}}e^{-\frac{(i+1)(4k-3)}{2(n-i-1)}}}{e^{n-k}}e^{-\frac{2(k-1)^2}{n-i-1}-\frac{1}{6(n-k-i)}}\\
&= & \frac{2^{2n-2i-3/2}n^{n+k-1}}{e^{n}} e^{-\frac{2(k-1)^2}{n-i-1}-\frac{(i+1)(4k-3)}{2(n-i-1)}-\frac{1}{6(n-k-i)}-\frac{k}{2(n-k)}}\\
&\geq & \frac{2^{2n-2i-3/2}n^{n+k-1}}{e^{n}} \left(1
 -\frac{2(k-1)^2}{n-i-1}-\frac{(i+1)(4k-3)}{2(n-i-1)}-\frac{1}{6(n-k-i)}-\frac{k}{2(n-k)} \right)\\
&\geq & \frac{2^{2n-2i-3/2}n^{n+k-1}}{e^{n}} \left(1 -
\frac{i(4k-3)+(2k-1)^2}{2(n-i-1)}
 - \frac{7k+i+2}{6n}\right),
\end{eqnarray*}
where the last inequality is derived from 
the following facts:
\[\frac{k}{2(n-k)}\leq \frac{k}{n} \]
for any $k<n/2$, and 
\[ \frac{1}{n-i-1}\leq \frac{i+2}{n}\]
for any $i$ and $k$ such that $k+i+1<n$.
\end{proof}

\subsection*{C2. Proof of Main Result}

Recall that a {galled network} is a BPN in which the two parents of each reticulation node lie in  the same tree-component.

Let $\Sigma$ be a set of $n$ taxa. To obtain a lower bound, we consider tree-child networks on $\Sigma$ in which the parents of each reticulation node are found in the top component. We use ${\cal COM}(A, k)$ to denote the set of such networks with $k$ reticulations on a fixed set $A$ of n taxa.

Let $\Sigma$ be a set of $n$ taxa and $L=\{\ell_1, ..., \ell_k\}$ such that $\Sigma\cap L=\emptyset$. For each subset $\Sigma'\subseteq \Sigma$, we use $X(\Sigma', L)$ to denote the set of  networks in ${\cal COM}(\Sigma'\cup L, k)$ such that the children of the $k$ reticulation nodes are leaves $\ell_1, \ell_2, ..., \ell_k$. We also use ${\cal F}(\Sigma\setminus \Sigma', k)$ to denote the set of forests on $\Sigma\setminus \Sigma'$ each consisting of $k$ subtrees.

Clearly, a BPN on $\Sigma$ in ${\cal COM}(n, k)$ can be uniquely  obtained from a BPN in  $X(\Sigma', L)$  by replacing each leaf $\ell_i$ with a subtree of a forest in ${\cal F}(\Sigma\setminus \Sigma', k)$. Since each forest in ${\cal F}(\Sigma\setminus \Sigma', k)$ contains  at least $k$ taxa, 
$1\leq i=\vert \Sigma'\vert \leq n-k$.  Therefore, 
By Propositions~\ref{proposition4} and \ref{proposition5_complex}, we have: 
\begin{eqnarray*}
     \vert {\cal COM}(\Sigma, k)\vert &= & \sum^{n-k}_{i=1} {n\choose i} \frac{(2n-2i-k-1)!}{2^{n-i-k}(n-i-k)!(k-1)!}\frac{(2i+2k-2)!}{2^{i+k-1}(i-1)!} \;\;\;\;\; \mbox{($i$:  no. of taxa in $\Sigma'$)} \nonumber\\
    &= & \frac{n!}{2^{n-1}(k-1)!} \sum^{n-k}_{i=1} \frac{(2n-2i-k-1)!}{(n-i-k)!(n-i)!}\frac{(2i+2k-2)!}{(i-1)!i!} \nonumber \\ 
     &= & \frac{n!}{2^{n-1}(k-1)!} \sum^{n-k-1}_{j=0} \frac{(2n-2j-k-3)!}{(n-j-k-1)!(n-j-1)!}\frac{(2j+2k)!}{(j)!(j+1)!} \;\;\;
     \mbox{(set $i=1+j$)} \nonumber \\  
     &= & {n\choose k}\frac{k}{2^{n-1}} \sum^{n-k-1}_{b=0} \frac{(2b+k-1)!}{(b)!(b+k)!}\frac{(2n-2b-2)!(n-k)!}{(n-b-k-1)!(n-b-k)!},
     \;\;\;\;\;\;\;\;\; (E35)  \nonumber
\end{eqnarray*}
where the last equality is obtained by setting
 $b=n-1-j-k$, which is the number of taxa in the forest minus $k$.
\\

\noindent {\bf Theorem C1}
{\it  
Let  ${\cal COM}(\Sigma, k)$ be the set of such networks with $k$ reticulations on a fixed set $\Sigma$  of n taxa in which the two parents of every reticulation node are in the top component. Then,
for any $k= o(n^{1/3})$, 
\begin{eqnarray*}
  \vert {\cal COM}(\Sigma, k) \vert \geq 
  {n\choose k} \frac{2^{n+k-1/2}n^{n+k-1}}{e^{n}}\left(1 - o(1)\right).
 \end{eqnarray*}
 }
\begin{proof} 
By 
(E35) and Lemma C3, 
\begin{eqnarray*}
  \vert {\cal COM}(n, k) \vert &\geq &
   {n\choose k}\frac{k}{2^{n-1}} \sum^{\lfloor n^{2/3} \rfloor}_{b=0} \frac{(2b+k-1)!}{(b)!(b+k)!}\frac{(2n-2b-2)!(n-k)!}{(n-b-k-1)!(n-b-k)!}\\
   &\geq & {n\choose k}\frac{2^{n-1/2}n^{n+k-1}k}
   {e^n}\left(\sum^{\lfloor n^{2/3} \rfloor}_{b=0} \frac{(2b+k-1)!}{(b)!(b+k)!}\frac{1}{2^{2b}}\right)
   \left(1-O(n^{-2/3})-O(n^{-1/3})\right).
\end{eqnarray*}

By Lemma~C1 and Lemma~C2, given above,
we further have:

\begin{eqnarray*}
  \vert {\cal COM}(n, k) \vert 
  &\geq & {n\choose k}\frac{2^{n-1/2}n^{n+k-1}k}
   {e^n}\left(\frac{2^k}{k}-\sum^{\infty}_{b=\lfloor n^{2/3} \rfloor+1 } \frac{(2b+k-1)!}{(b)!(b+k)!}\frac{1}{2^{2b}}\right)
   \left(1-O(n^{-1/3})\right)\\
    &\geq & {n\choose k}\frac{2^{n-1/2}n^{n+k-1}k}
   {e^n}\left(\frac{2^k}{k}-\frac{2^k}{k}\frac{k}{\sqrt{\pi}n^{1/3}}\right)
   \left(1-O(n^{-1/3})\right)\\
    &= & {n\choose k}\frac{2^{n+k-1/2}n^{n+k-1}}
   {e^n}\left(1-\frac{k}{\sqrt{\pi}n^{1/3}}\right)
   \left(1-O(n^{-1/3})\right)
\end{eqnarray*}
This concludes the proof.
\end{proof}


\end{document}